\documentclass[lineno]{biometrika}

\usepackage{amsmath}

\usepackage{multirow, bbm, xcolor, subcaption, caption, tabularx,float}
\usepackage{algpseudocode}
\usepackage{bbm}
\usepackage{amsfonts}
\usepackage{amssymb}
\usepackage{multirow}
\usepackage{chngcntr}
\usepackage{apptools}
\usepackage{dsfont}

\def\bE{\mathbb E}
\def\bE{E}
\def\R{\mathbb R}
 \newcommand{\xhdr}[1]{\vspace{0.1mm}\noindent{{\bf #1.}}}
\def\P{{\text{pr}}}

\def\argmin{\mathop{\rm arg\,min}}

\def\diag{\mathop{\rm diag}\nolimits}

\def\dim{\mathop{\rm dim}\nolimits}

\def\sign{\mathop{\rm sign}}

\def\supp{{\rm supp}}

\def\Sum{\overset{n}{\underset{i=1}{\sum}}}

\usepackage{newtxtext}
\usepackage[subscriptcorrection]{newtxmath}

\graphicspath{{./art/}}

\usepackage[plain,noend]{algorithm2e}

\makeatletter
\renewcommand{\algocf@captiontext}[2]{#1\algocf@typo. \AlCapFnt{}#2} 
\def\@algocf@capt@plain{top}
\renewcommand{\algocf@makecaption}[2]{%
  \addtolength{\hsize}{\algomargin}%
  \sbox\@tempboxa{\algocf@captiontext{#1}{#2}}%
  \ifdim\wd\@tempboxa >\hsize
    \hskip .5\algomargin%
    \parbox[t]{\hsize}{\algocf@captiontext{#1}{#2}}
  \else%
    \global\@minipagefalse%
    \hbox to\hsize{\box\@tempboxa}
  \fi%
  \addtolength{\hsize}{-\algomargin}%
}
\makeatother

\begin{document}



\markboth{C. Donnat, O. Klopp and N. Verzelen.}{One-Bit Total Variation Denoising over Networks}

\title{One-Bit Total Variation Denoising over Networks with Applications to Partially Observed Epidemics}

\author{Claire DONNAT}
\affil{Department of Statistics, University of Chicago,\\ 5747 S Ellis Ave, Chicago, Illinois, 60637, United States of America
\email{cdonnat@uchicago.edu}}

\author{Olga KLOPP}
 \affil{
 ESSEC Business School\\
 3 Av. Bernard Hirsch, 95000 Cergy, France\\CREST, ENSAE, Institut Polytechnique de Paris,\\ 5 Av. Le Chatelier, 91120 Palaiseau, France
\email{kloppolga@math.cnrs.fr}}

\author{\and Nicolas VERZELEN}
\affil{MISTEA, Université Montpellier, INRAE, Institut Agro\\
2 Pl. Pierre Viala, 34000 Montpellier, France\\
\email{nicolas.verzelen@inrae.fr
}}

\maketitle

\begin{abstract}
	This paper introduces a novel approach for epidemic nowcasting and forecasting over networks using total variation (TV) denoising, a method inspired by classical signal processing techniques. 
 Considering a network that models a population as a set of $n$ nodes characterised by their infection statuses $Y_i$ and that represents contacts as edges, we prove the consistency of graph-TV denoising for estimating the underlying infection probabilities $\{p_i\}_{ i \in \{1,\cdots, n\}}$ in the presence of Bernoulli noise. 
 Our results provide an important extension of existing bounds derived in the Gaussian case to the study of binary variables --- an approach hereafter referred to as one-bit total variation denoising. 
 The methodology is further extended to handle incomplete observations, thereby expanding its relevance to various real-world situations where observations over the full graph may not be accessible. Focusing on the context of epidemics, we establish that one-bit total variation denoising enhances both nowcasting and forecasting accuracy in networks, as further evidenced by comprehensive numerical experiments and two real-world examples. 
 The contributions of this paper lie in its theoretical developments, particularly in addressing the incomplete data case, thereby paving the way for more precise epidemic modelling and enhanced surveillance strategies in practical settings.
\end{abstract}


\section{Introduction}

We consider the problem of signal denoising over networks
and apply our results to the analysis of infectious disease transmission over contact networks.
The spread of epidemics over networks is an important area of research with implications for public health \citep{ottaviano2018optimal,sanchez2022multilayer}, epidemiology \citep{ danon2011networks,keeling2005networks,moreno2002epidemic,shirley2005impacts, spricer2019sir}, and network analysis \citep{kim2021network,wan2014inferring}. While the study of disease propagation through interconnected populations had already gained significant attention over the past decade \citep{keeling2005implications, newman2002spread,pastor2015epidemic,pellis2015eight}, the recent emergence of new infectious diseases, such as COVID-19, further underscores the importance of developing comprehensive models that can capture the intricate interplay between epidemic transmission processes and the underlying network structures through which they spread \citep{keeling2005networks, manzo2020complex}.

Studies of epidemic processes typically rely on compartmental models \citep{kermack1927contribution, tolles2020modeling}, which split a population of size $n$ into different compartments (or groups) based on their disease status, and specify how agents transition between disease states (e.g., Susceptible, Infectious and Recovered). These approaches can be further refined into two categories, depending on how they account for randomness and uncertainty in the disease transmission process:  \textit{deterministic}, versus \textit{stochastic} epidemic models. Deterministic models \citep{bowman2005mathematical, carcione2020simulation, eikenberry2020mask, kumar2021infection, mubarak2021stochastic} assume that the disease transmission process can be precisely described using mathematical equations with fixed parameters. These models typically use differential equations to describe how the numbers of infected, susceptible, and recovered individuals change over time, and do not capture individual-level variability in the disease transmission process.
 Although deterministic models have provided valuable insights into the dynamics of epidemics, they often fail to capture the inherent randomness and unpredictability associated with real-world outbreaks. By contrast, the consideration of variability recognises that epidemic processes are influenced by random events, chance encounters, and unforeseen interactions among individuals. This acknowledgment prompts the need for stochastic models that can explicitly incorporate randomness into the modeling framework \citep{ball_2019}. 
 
 Instead of using differential equations, stochastic models often rely on techniques such as agent-based modelling \citep{hackl2019epidemic,perez2009agent,siettos2015modeling,teweldemedhin2004agent} or stochastic differential equations on compartmental models \citep{andersson2012stochastic, britton2010stochastic, bu2022likelihood, huang2022detecting, morsomme2022exact} to simulate individual interactions and disease transmission events. In this framework, each interaction or transmission event is subject to chance. The fluxes between compartments therefore follow a probability distribution, rather than a deterministic equation. Although these approaches are designed to better capture the uncertainty inherent in epidemic spreads, most of them only typically incorporate randomness in terms of transmission events: given a contact between an infected and susceptible individual, the transmission is subject to chance, and typically modelled as a Bernouilli variable. Specifically,  these models do not attempt to remedy another major limitation of compartmental models: the uniform mixing assumption. This assumption implies that every infected individual in the population can potentially transmit the disease to any susceptible individual in the network, which is not always realistic \citep{keeling2005networks}. Such a simplification overlooks critical dynamics in epidemic spread on networks, particularly in large populations where the uniform mixing assumption is known to fail.  These approximations can significantly affect the accuracy of models that attempt to capture how diseases spread through complex networks.

A possible solution to address this limitation is to model transmissions within a given contact network, thereby capturing more realistically possible transmission events.  The structure of contact networks can in fact known to significantly impact disease spread \citep{ganesh2005effect}. Indeed, while the standard epidemic models described in the previous paragraph correspond to fully connected networks,  accounting for a more realistic structure of interactions and connections between individuals may lead to more accurate predictions and more effective control strategies \citep{chakrabarti2008epidemic,wang2013effect}. 
  For example, \citet{pastor-satorras_epidemic_2001} have shown that the structure of the contact network greatly influences the epidemic threshold and disease prevalence. Research on temporal networks \citep{holme_temporal_2012} has shown that the timing of interactions plays a crucial role in shaping epidemic dynamics by affecting disease transmissibility. Epidemic models on graphs have also highlighted the importance of super-spreaders, that is,  individuals or nodes with a disproportionately high potential for transmission. Finally, \citet{keeling_networks_2005} introduced the notion of ``assortativity" in networks, where highly connected nodes preferentially interact with each other. This assortative mixing has significant implications for the effective targeting interventions and controlling epidemics.

Another key limitation of traditional deterministic and stochastic differential equation models is the assumption that the current state of the epidemic is known: all agents' infectious, susceptible, or recovered states are observed.  Knowledge of the current state of the epidemic is, in fact, essential for making accurate predictions about its future course, as small inaccuracies can lead to substantial changes in epidemic trajectories over time \citep{siegenfeld2020models}. However, accurately assessing the current situation of an epidemic, a process often referred to as nowcasting, is notoriously challenging \citep{desai2019real, mcgough2020nowcasting, wu2021nowcasting}, due to the availability of quality data and the delays in reporting of current cases~\citep{chakraborty2022nowcasting, rosenfeld2021epidemic}. 
In many situations, complete information about individuals' health status cannot be disclosed due to privacy concerns. It thus becomes essential to develop methods that are adapted to the practical constraints of real data and respect individuals' rights while still providing valuable insights into epidemic spread.

%
 
 In the present paper, we propose a new approach to nowcasting and forecasting of epidemic spreads on a network inspired by total-variation signal denoising. 
 Signal denoising based on the total variation (TV) penalty is a widely used technique in signal processing and image reconstruction that aims to remove noise from an observed signal while preserving important spatial or structural features. This approach uses the TV penalty, which measures the variation in the signal over a network. Letting $p_i, i \in \{1, \cdots, n\}$  be a signal over a network $\mathcal G$ on $n$ nodes and $E$ edges, the total variation penalty can be written as: $\sum_{(i,j) \in E} | p_i-p_j|.$ By penalising differences in the signal over neighbouring vertices, this term acts as a regularization penalty that enforces smoothness in the reconstructed signal. TV regularization is typically formulated as an optimisation problem, where the goal is to minimise the TV of the signal subject to a fidelity constraint in the form of a data-fitting loss, ensuring the reconstructed signal remains close to the original data. A more extensive review of Total-variation denoising is provided in Section~\ref{sec:related_works}.

In this paper, we study TV denoising in the setting of Bernoulli observations and show that this approach improves nowcasting estimates and prediction accuracy of dynamic behavior of epidemic processes over simple compartmental approaches.  To address the challenges of epidemic nowcasting, we also extend our approach to the case of incomplete observations which allows a better assessment of the disease's spread and can help implement more targeted surveillance and control measures.

\subsection*{Contributions}

Our contributions are threefold and summarised as follows.
First, we prove that the total variation denoiser that we propose -- thereafter referred to as one-bit TV denoiser --- is consistent in the case of Bernoulli noise (see Theorem \ref{thm:denoising}). These results extend existing bounds that have been obtained in the literature in the case of Gaussian noise \citep{hutter_optimal_2016}. In particular, we obtain new bounds on the risk of TV denoiser measured in terms of the $\ell_1$-norm.
 Second, we extend our analysis to the challenging case of missing observations and demonstrate good performance of the one-bit TV denoiser (see Theorem \ref{thm:denoising_missing}). To the best of our knowledge, such results are new in the literature. Finally,  we show how the one-bit total variation denoiser can be applied to improve the prediction of virus spread dynamics on networks (see Section \ref{sec:epidemics}). Our results are supported by extensive numerical studies that confirm the advantages of the proposed method, as well as two real-data experiments.

\subsection*{Notations}
For any $k$, $\mathbf{1}_{k}$ denotes the $k$-dimensional vector of ones and 
$\mathbf{0}_{k}$ denotes the $k$-dimensional vector of zeros.
For an integer $k\in \mathbb{N}$, we write $[k]=\{1,\dots,k\}$. We denote by $A^{\dagger}$ the Moore-Penrose pseudo-inverse of a matrix $A$ and by $\Vert A\Vert$ its operator norm.
Let $S\subset [n]$ be a subset of the nodes in $[n]$, we will use the notation $d_{\max}(S)$ to denote the maximum degree of nodes $i$  in subset $S$.

\section{Signal denoising}\label{Sec:signal_denoising}
Consider a network  where each node $i \in V$
 has an underlying state $y_i\in \{0, 1\}$ where
$0$ and $1$ may respectively represent, for example, healthy and infected states.
The
probability that person $i$ is in the infected state (or the proportion of the infected subpopulation at node $i$) is
indicated by $p^{*}_i$, that is, $p^{*}_i=\P(y_i=1)$. Let $p^{*}=(p^{*}_i)_{i=1,\dots n}$ denote the vector that captures the infected states of an interconnected population consisting of $n$ nodes.
The state of the network can thus be represented by the following ``signal + noise" model:
\begin{equation}\label{model}
	Y=p^{*}+\xi
\end{equation}
where $Y=(y_i)_{i\in[n]}$ are our observations and $\xi$ is a vector of centred Bernoulli noise. Given one vector of observations $Y$,  we consider the problem of estimating the underlying $(p^{*}_i)_{i\in[n]}$.

Let $\mathcal G=(V, E)$ be the underlying
undirected connected graph with vertex set $V$ and edge set $E$.  We set the cardinalities of both sets as $\vert V\vert =n$ and $\vert E\vert =m$.  We denote by  $A = (a_{ij})_{(i,j)\in [n]\times [n]}$  the graph's adjacency matrix.  It will be convenient for us to represent a graph by its edge-vertex incidence matrix $D\in \{-1,0,1\}^{m\times n}$. To each edge $e=(i,j)\in E$ corresponds a row $D_{e}$ of $D$ where the $k$th entry $D_{e,k}$ of $D_{e}$ is given by
\begin{equation*}
	D_{e,k}=\begin{cases}
		1& \text{if } k=\min(i,j)\\
		-1& \text{if } k=\max(i,j)\\
		0& \text{otherwise.}\\
	\end{cases}
\end{equation*}

Note that that  $L=D^{T}D=diag(A\mathbf{1}_{n})-A$ is the \textit{unnormalized Laplacian} of the graph $\mathcal G$. Here $diag(A\mathbf{1}_{n})$ is the diagonal matrix with $j$th diagonal element given by the degree of vertex $j$.

%
 The TV denoiser $\widehat p$ associated to $G$ is any solution of the following minimization problem:
\begin{equation}\label{def:tv_denoiser}
	\widehat{p}\in\underset{p\in \mathbb{R}^{n}}
	{\argmin}\left \{\dfrac{1}{n}\Sum \left (y_i-p_i\right )^{2}+\lambda \Vert Dp\Vert_1\right \},
\end{equation}
where $\lambda>0$ is a regularization parameter and the TV penalty $\Vert Dp\Vert_1$ is the convex relaxation of $\Vert Dp\Vert_0$, which corresponds to the number of times $p$ changes values along the edges of the graph $G$. Note that \eqref{def:tv_denoiser} is a convex problem that may be solved efficiently \citep{arnold_efficient_2016}. Given a $p\in \mathbb{R}^{n}$, let $\mathcal{P}(p)=\left (\mathcal{P}(p_{i})\right )^{n}_{i=1}\in [0,1]^{n}$ be defined as follows
\begin{equation*}
	\mathcal{P}(p_{i})=\begin{cases}
		p_i& \text{if }\; p_i\in [0,1],\\
		0& \text{if } \;p_i<0,\\
		1& \text{if } \;p_i>0.\\
	\end{cases}
\end{equation*}
Let $\mathcal{F}=\dfrac{1}{n}\Sum \left (y^{k_0}_i-p_i\right )^{2}+\lambda \Vert Dp\Vert_1$. We have that for any $p$, $\mathcal{F}(\mathcal{P}(p))\leq \mathcal{F}(p)$. So, without loss of generality, we assume that $\widehat{p}_{i}\in [0,1]^{n}$.

We start by defining the \textit{compatibility factor} and the \textit{inverse scaling factor} that will be crucial in characterising the performance of the TV denoiser:
\begin{definition}[Compatibility factor]\label{def:compatibility factor}
	Let $D\in\{-1,0,1\}^{m\times n}$ be the incidence matrix. The compatibility factor of a set $T\subset [m]$ is defined as
	\[\kappa_{T}=\kappa_{T}(D)=\underset{p\in [0,1]^{n}}{\inf}\dfrac{\sqrt{\vert T\vert}\,\Vert p\Vert_{2}}{\left \Vert \left (Dp\right )_{T}\right \Vert_{1}}\quad\text{for}\quad T\not = \emptyset.\]  
\end{definition}
\begin{definition}[Inverse scaling factor]
	Let $S=D^{\dagger}=[s_{1},\dots,s_{m}]$. The inverse scaling factor of $D$ is defined as 
	\[\rho=\rho(D)=\underset{j\in[m]}{\max}\Vert s_j\Vert_{2}.\]
\end{definition}
The following result shown in \citet{hutter_optimal_2016} provides a bound on both $\rho$ and $\kappa_{T}$:
\begin{proposition}\label{prp: bounds_kappa_rho}
	Let $D$ be the incidence matrix of a connected graph $G$ with maximal degree $d_{\max}$ and $\emptyset\not = T\subseteq E$. Let $0=\lambda_{1}< \lambda_{2} \leq \dots \leq \lambda_{n}$ be the eigenvalues of the Laplacian $L$.  Then,
	\begin{equation*}
		\kappa_{T}\geq \dfrac{1}{2\min\{\sqrt{d_{\max}},\sqrt{\vert T\vert}\}}\quad \text{and}\quad \rho \leq \dfrac{\sqrt{2}}{\lambda_{2}}.
	\end{equation*}

\end{proposition} 
Equipped with these notations, we establish the following bound on the risk of TV denoiser \eqref{def:tv_denoiser}: 
\begin{theorem}[Risk bound for one-bit TV denoising]\label{thm:denoising}
	Fix $\delta\in[0,1]$, $T\subset [m]$ and assume that $G$ is connected. Define the regularization parameter 
	\begin{equation}\label{def: lambda}
		\lambda=\frac{\sqrt{2}\rho}{n}\log\left (\frac{4n^{2}}{\delta}\right )\enspace . 
	\end{equation}
	Then, the TV denoiser $ \widehat{p}$ defined in \eqref{def:tv_denoiser} satisfies 
	\begin{equation} \label{bound_l2_risk}
			\begin{split}
				{\Vert \widehat p-p^{*}\Vert^{2}_{2}}\leq  \dfrac{16\rho^{2}\vert T\vert\log^{2}\left (\frac{4n^{2}}{\delta}\right )}{\kappa^{2}_{T}}+4\sqrt{2}\rho\left \Vert \left (Dp^{*}\right )_{T^{c}}\right \Vert_{1}\log\left (\frac{4n^{2}}{\delta}\right )
				 + \frac{4\Vert p^{*}\Vert_{0}\log\left (\tfrac{4}{\delta}\right )}{n}
			\end{split}
		\end{equation}
		with probability at least $1-\delta$.	
	\end{theorem}
 The proof of Theorem~\ref{thm:denoising} is provided in Appendix~\ref{appendix:thm:denoising}. 
Let $\mathcal S=\supp( p^{*})$ and $s=\vert \mathcal S \vert= \|p^*\|_{0}$. Inequality \eqref{bound_l2_risk} allows trading off $\vert T\vert$ and $\left \Vert \left (Dp^{*}\right )_{T^{c}}\right \Vert_{1}$. Taking $T=\supp(Dp^{*})$ , we obtain the following bound on the estimation  error:
\[
	{\Vert \widehat p-p^{*}\Vert^{2}_{2}}\leq  \dfrac{16\rho^{2}\Vert Dp^{*}\Vert_{0}\log^{2}\left (\frac{4n^{2}}{\delta}\right )}{\kappa^{2}_{T}}+ \frac{4s\log\left (\tfrac{4}{\delta}\right )}{n}.
\]

On the other hand, taking $T$ to be the empty set in \eqref{bound_l2_risk}, we obtain the alternative error bound:
\[
{\Vert \widehat p-p^{*}\Vert^{2}_{2}}\leq  4\sqrt{2}\rho\left \Vert Dp^{*}\right \Vert_{1}\log\left (\frac{4n^{2}}{\delta}\right )+ \frac{4s\log\left (\tfrac{4}{\delta}\right )}{n}.
\]
The next proposition provides a risk bound in $\ell_1$ norm.
\begin{proposition}\label{prp:risk_l1}
	Fix $\delta\in[0,1]$, $T\in[m]$ and assume that $G$ is connected. Define the regularization parameter 
	\begin{equation*}
		\lambda=\frac{\sqrt{2}\rho}{n}\log\left (\frac{4n^{2}}{\delta}\right )
	\end{equation*}
	Then, the TV denoiser $ \widehat{p}$ defined in \eqref{def:tv_denoiser} satisfies
	\begin{equation} \label{bound_l1_risk}
		\begin{split}
			{\Vert \widehat p-p^{*}\Vert_{1}}&\leq  \dfrac{4\rho\sqrt{s\vert T}\vert}{\kappa_{T}}\log\left (\frac{4n^{2}}{\delta}\right )+2s\sqrt{\frac{\log\left (\frac{4}{\delta}\right )}{n}}\\&
			+ 3\sqrt{\rho\,s\left \Vert \left (Dp^{*}\right )_{T^{c}}\right \Vert_{1}\log\left (\frac{4n^{2}}{\delta}\right )} + \sqrt{2}\rho \Vert Dp^*\Vert_0 \log\left (\frac{4n^{2}}{\delta}\right )\enspace  , 
		\end{split}
	\end{equation}
	with probability at least $1-\delta$. Here, we recall that $s$ stands for $\|p^*\|_0$. 
\end{proposition}
 The proof of Proposition~\ref{prp:risk_l1} is provided in Appendix~\ref{appendix:prp:risk_l1}. In Proposition \ref{prp:risk_l1}, we can again trade off $\vert T\vert$ and $\left \Vert \left (Dp^{*}\right )_{T^{c}}\right \Vert_{1}$. Taking $T=\supp(Dp^{*})$, 
 the previous bound becomes: 
\[
{\Vert \widehat p-p^{*}\Vert_{1}}\leq \left(\dfrac{4\sqrt{s\Vert Dp^*\Vert_0}}{\kappa_{T}}+\sqrt{2}\Vert Dp^*\Vert_0\right)\rho\log\left (\frac{4n^{2}}{\delta}\right )+2s\sqrt{\frac{\log\left (\frac{4}{\delta}\right )}{n}}.
\]
On the other hand, taking $T$ to be the empty set, we obtain:
\[
{\Vert \widehat p-p^{*}\Vert_{1}}\leq  
2s\sqrt{\frac{\log\left (\frac{4}{\delta}\right )}{n}}+ 3\sqrt{\rho\,s\left \Vert Dp^{*}\right \Vert_{1}\log\left (\frac{4n^{2}}{\delta}\right )} + \sqrt{2}\rho \Vert Dp^*\Vert_0 \log\left (\frac{4n^{2}}{\delta}\right ).
\]

	The errors terms given by \eqref{bound_l2_risk} and \eqref{bound_l1_risk} crucially depend on the inverse scaling factor $\rho$ and the compatibility factor $\kappa_{T}$. \citet{hutter_optimal_2016} provides upper bounds for these quantities for different graphs. Using $\lesssim$ for an inequality where the left-hand side is bounded by the right-hand side up to a numerical constant, we specialize the upper bound on the risk of the TV denoiser to specific graphs. All the following bounds ( case 1 through 4) hold with probability at least $1-4/n$.

\xhdr{Case 1: $2D-$grids}  For the 2D grid, $\rho \lesssim \sqrt{\log n}$ (see  \citet{hutter_optimal_2016}). The $\ell_2$ and $\ell_1$ bounds on the risk of Proposition~\ref{prp:risk_l1} and Theorem~\ref{thm:denoising} thus yield that:
		\[
		{\Vert \widehat p-p^{*}}\Vert_{2}\lesssim \left \{\sqrt{\left \Vert Dp^{*}\right \Vert_{1} +\frac{s}{n\log^{1/2}(n)}} \;\right \}\log^{3/4}\left (n\right )\quad \text{and}
		\]
		\[
		{\Vert \widehat p-p^{*}\Vert_{1}}\lesssim \left \{\sqrt{\frac{s \left \Vert Dp^{*}\right \Vert_{1}}{\log^{3/2}(n)}}+\left \Vert Dp^{*}\right \Vert_{0}\right \}\;\log^{3/2}\left (n\right ).
		\]
\xhdr{Case 2: Complete graphs} For  a complete graph,  $\rho \lesssim 1/n$ and $\kappa_{T} \gtrsim 1/\sqrt{n}$ \citep{hutter_optimal_2016}. This implies:	
		\[
		{\Vert \widehat p-p^{*}\Vert_{2}}\lesssim \sqrt{\dfrac{\left \Vert Dp^{*}\right \Vert_{1} +s}{n}}\;\log^{1/2}\left (n\right )\quad \text{and}
		\]
		\[
		{\Vert \widehat p-p^{*}\Vert_{1}}\lesssim 
	\left (	\sqrt{\dfrac{s \left \Vert Dp^{*}\right \Vert_{1} }{n\log(n)}}+\frac{\Vert Dp^*\Vert_0} {n}\right )\log\left (n\right ).
		\]
\xhdr{Case 3: Star graphs} For a star graph, $\rho \lesssim 1$ and $\kappa_{T} \gtrsim 1/\sqrt{\vert T\vert}$, so that:
		\[
		{\Vert \widehat p-p^{*}\Vert_{2}}\lesssim \left \{ \sqrt{\left \Vert Dp^{*}\right \Vert_{1}+\frac{s}{n}}\right \}\log^{1/2}\left (n\right )\quad \text{and}
		\]
			\[
		{\Vert \widehat p-p^{*}\Vert_{1}}\lesssim \left (\sqrt{\frac{s\left \Vert Dp^{*}\right \Vert_{1}}{\log(n)}}+ \Vert Dp^*\Vert_0 \right )\log\left (n\right )
		\]
\xhdr{Case 4: Random graphs}  
  Random $d_n$-regular and  Erd\"os-R{\'e}nyi graphs $\mathcal{G}(n,p)$ with $p=d_n/n$
		where $d_n=d_0(\log n)^{\beta}$ for some $\beta >0$, $d_0>1$ exhibit a spectral
gap of the order $O(d_n)$
(see \cite{Kolokolnikov2014,friedman08}) which implies:
		\[
		{\Vert \widehat p-p^{*}\Vert_{2}}\lesssim \sqrt{\frac{\Vert Dp^{*}\Vert_{1}\log\left (n\right )}{d_n} + \frac{s\log(n)}{n}}\quad \text{and}
		\]
		\[
		{\Vert \widehat p-p^{*}\Vert_{1}}\lesssim \left (\sqrt{\frac{s \Vert Dp^{*}\Vert_{1}}{d_n}}+\frac{\Vert Dp^*\Vert_0 \sqrt{\log\left (n\right )}}{d_n}+\frac{s}{\sqrt{n}}\right )\sqrt{\log\left (n\right )}.
		\]

 \xhdr{Discussion} These bounds allow to shed more light on the interplay between the size of the support $\supp(Dp^*)$, the size of the epidemic $\| p^*\|_0$, as well as the topology of the graph. At a high level --- and as confirmed by our experiments in Section~\ref{sec:num_study}---, we expect our total variation denoiser to be better suited to situations where the cardinality of the signal's support is small compared to the size of the graph ($\|p^*\|_0=s <n$), and the maximal degree of the graph remains of constant order, so that  $|\supp(Dp^*)| \leq s {d_{\max}}(V)$ is also controlled. This means for instance that the denoiser is likely to provide a better control of the $\ell_1$-error in the $2D$-grid  (where the $\ell_1$ error is of the order of $s \log^{3/2}(n)$ compared to star graph, where the error could be in the order of $sn \log(n)$. 

\section{ Virus spread dynamics on network} \label{sec:epidemics}

We now apply our one-bit total variation denoiser to the analysis of epidemic spreads over networks.
As highlighted in the introduction, numerous compartmental models have been suggested to depict the spread of epidemic processes over networks. We refer the reader to  the recent surveys by \citet{draief_massoulie_2009,pastor_epidemic_2015} or \citet{pare_modeling_2020} for a more extensive review of these methods. We first propose deploying our model to the simple Susceptible - Infectious - Susceptible (SIS), before showing how our approach can be used in a wider variety of epidemic models.

\subsection*{SIS Model}

We start by considering the networked version of the susceptible - infectious - susceptible (SIS) compartmental model, originally introduced in \citet{kermack1927contribution}.  This widely used model has been indeed adapted for network settings in various subsequent studies \citep{ahn_global_2013, pare_analysis_2020}, as detailed in the subsequent paragraph.
We take the vector $(p_i)_{i \in [n]}$ to represent the probabilities of each node being infected, with each entry corresponding to an individual node's infection probability. In the networked setting \citep{pare_modeling_2020}, the evolution of the $p_i$s is governed by the following differential equation:
\begin{equation}\label{SIS_ODE}
\dot{p}_i(t)= (1-p_i(t))\beta_i \sum_{j=1}^{n}\omega_{ij}p_j(t)-\gamma_ip_i(t).	
\end{equation}
Here $\omega_{ij}$ quantifies the strength of the connection from node $j$ to node $i$, $\beta_i>0$ is the infection (or susceptibility) rate and $\gamma_i$ is the healing rate. The \textit{infection rate} $\beta_i$ quantifies the rate at which a susceptible individual (node $i$) in the population is infected through contacts with infected individuals (denoted here by the sum $\sum_{j}w_{ij}p_j$). In many cases, the susceptibility to a virus may not be uniform across the population. Susceptibility to a virus often varies across a population due to factors such as age, pre-existing health conditions, or immunity from previous exposure or vaccination. To account for this heterogeneity in infection rates, one solution consists in identifying relevant subgroups within the population and determining the infection rate for each subgroup.

The \textit{healing rate} $\gamma_i$ models the rate at which an infected individual $i$ transitions between infected and recovered states. When $\gamma_i$ is assumed to be identical across all individuals ($\gamma_i = \gamma, \quad \forall i=1, \cdots, n$), it can be understood as the proportion of infected individuals who recover from the virus in a given time window $\Delta t$. In this case, the healing rate can be estimated by tracking the number of individuals who transition from infected to recovered status over time, and calculated as the number of recoveries divided by the number of confirmed cases during that time window $\Delta t$. Usually, the time window for calculating the recovery rate is determined by the average duration of infection.  In general, to estimate the healing rate, we need information on the typical progression of the disease, including the duration of the infection, the severity of symptoms, and the factors that contribute to recovery. 
Note that the healing rate can vary depending on factors such as age, underlying health conditions, and the availability and effectiveness of treatments.

The SIS model introduced in \eqref{SIS_ODE} can be derived from a subpopulation perspective or through a mean-field approximation of a $2^{n}$ state Markov chain model  \citep{pare_analysis_2020}. As the epidemics dynamic is usually sampled at discrete time steps, we will consider the discrete version of model \eqref{SIS_ODE}:
\begin{equation}\label{SIS_DT}
	p^{k+1}_i= p^{k}_i +  (1-p^{k}_i)\beta_i \sum_{j=1}^{n}\omega_{ij}p^{k}_i-\gamma_ip^{k}_i	
\end{equation}
where $k>0$ is the time index. In matrix form, the model can be re-written as:
\begin{equation}\label{SIS_Matrix}
	p^{k+1}= p^{k} +  \left [(1-P^{k})B \Omega -\Gamma \right] p^{k}
\end{equation}
where $P^{k}=\diag{p^{k}}$, $B=\diag{\beta_i}$, $\Omega=(\omega_{ij})$ and $\Gamma=\diag{\gamma_i}$. We denote by 
$$\mathcal{O}_{k}=1 +  (1-P^{k})B \Omega -\Gamma. $$ 
For this model to be well defined,  we need the following assumption:
\begin{assumption}\label{Ass_SIS_model}
	For all $i\in[n]$, we have $\gamma_i<1$ and $\sum_{j=1}^{n}\beta_i\omega_{ij}<1$.
\end{assumption}
In this case, we have the following result:
	\begin{lemma}[\citep{pare_analysis_2020}]
	Consider the model \eqref{SIS_Matrix} under Assumption \ref{Ass_SIS_model}.
	Suppose $p^{0}_i\in[0,1]$ for all $i\in[n]$. Then, for all $k>0$ and $i\in[n]$, $p^{k}_i\in[0,1]$.
\end{lemma}


Using data on healthy / infected individuals at time point $k^{0}$, we can estimate ${p}^{k_0}$ using the TV estimator given by \eqref{def:tv_denoiser}. Let \begin{equation}\label{est_evolution_operator}
	\widehat{\mathcal{O}}_{k}=1 +  (1-\widehat P^{k}) B \widehat \Omega - \Gamma
\end{equation} 
 be the evolution operator, where $\widehat P^{k}=\diag(\widehat{p}^{k})$ and $\widehat{\Omega}$ is an estimator of $\Omega$ from noisy observations of the network (e.g. using variational maximum likelihood from some (partial) observations of the network \citep {gaucher_maximum_2021}, as discussed in Remark 3 below).  
Then, by recursively computing $\hat p^{k}= \widehat{\mathcal{O}}_{k-1}\hat p^{k-1} $ for $k>k^{0}$,  we can predict the evolution of the epidemic. In particular, the expected number of infected individuals at time $k$, $ \bE (I^{k})= \sum p^{k}_i=\Vert p^{k}\Vert_{l_1}$ can be estimated using Proposition \ref{prp:risk_l1}. 
\begin{proposition}\label{prp:est_number_infections}
	Consider the discrete time SIS given by \eqref{SIS_DT}. Assume that we observe $Y^{k_0}\in\{0,1\}^{n}$ which encodes the data on the healthy / infected status of a population of $n$  connected individuals at time point $k^{0}$. Fix $\delta\in[0,1]$, $T\subset [m]$ and assume that the contact graph $G$ is connected. 
	Then, with probability at least $1-\delta$, we have 
	\begin{equation*}
		\begin{split}
		&\left \vert 	\Vert 	\widehat p^{k}\Vert_{l_1} -\bE (I^{k})\right \vert\leq  \left (1-\gamma_{\min}+\beta_{\max} \Vert \Omega\Vert_{1,\infty}\right )^{k-k_0}\left (\dfrac{4\rho\sqrt{s\vert T}\vert}{\kappa_{T}}\log\left (\frac{4n^{2}}{\delta}\right )+\right .\\&
	\left .	+2s\sqrt{\frac{\log\left (\frac{4}{\delta}\right )}{n}}+
		ˇ3\sqrt{\rho\,s\left \Vert \left (Dp^{k_0}\right )_{T^{c}}\right \Vert_{1}\log\left (\frac{4n^{2}}{\delta}\right )} + \sqrt{2}\rho \Vert Dp^{k_0}\Vert_0 \log\left (\frac{4n^{2}}{\delta}\right )\right )\enspace . 
		\end{split}
	\end{equation*}
where $\gamma_{\min}=\underset{i}{\min}\;\gamma_i$ and $\beta_{\max}=\underset{i}{\max}\;\beta_i$ 
\end{proposition}
\begin{proof} We have $\bE (I^{k})=\Vert  p^{k}\Vert_{l_1}$ and
	\begin{equation*}
		\left \vert 	\Vert 	\widehat p^{k}\Vert_{l_1} -\Vert  p^{k}\Vert_{l_1}\right \vert\leq \Vert 	\widehat p^{k} - p^{k}\Vert_{l_1}\leq \left (1-\gamma_{\min}+\beta_{\max} \Vert \Omega\Vert_{1,\infty}\right )^{k-k_0}\Vert \widehat p^{k_0}-p^{k_0}\Vert_{l_1}.
	\end{equation*}
The result follows by bounding $\Vert \widehat p^{k_0}-p^{k_0}\Vert_{l_1}$ using the results of Proposition \ref{prp:risk_l1}. 
\end{proof}
Note that Assumption \ref{Ass_SIS_model} implies $\beta \Vert \Omega\Vert_{1,\infty}\leq 1$.

\begin{remark}
	Beyond SIS model, one can also consider Susceptible - Infected - Recovered (SIR) and Susceptible - Exposed - Infected - Recovered (SEIR) models \citep{pare_modeling_2020}. For a SIR process, for example, let $r^{k}_i$ denotes   the probability of node $i$ being  recovered at time $k$. Then, the following equations govern the evolution of $p^{k}$ and $r^{k}$ :
	\begin{equation*}
		\begin{split}
			p^{k+1}_i&= p^{k}_i +  (1-p^{k}_i-r^{k}_i)\beta_i \sum_{j=1}^{n}\omega_{ij}p^{k}_i-\gamma_ip^{k}_i\\
			r^{k+1}_i&= r^{k}_i -\gamma_ip^{k}_i.
		\end{split}
	\end{equation*}
\end{remark}
Our denoiser can therefore easily be deployed in the context of SIR and SEIR models by simply substituting the evolution operator of Equation~\ref{est_evolution_operator} with a appropriate version.

\begin{remark}[\textbf{Taking into account false positives}]
	In the case of epidemic models on networks, it is also necessary to account for errors that may occur in our observations $Y^{k_0}$   due to false positives. Such false positives might be due to errors due to imperfect test specificity for instance: while gold standard medical tests can have specificity up to 99\%, reporting based on symptom data is usually less accurate, and can thus lead to false positives.
 Here we assume that we know the false positive rate $\alpha$, and that our observations $y^{k}_{i}$ follow Bernoulli distribution with probability $\rho^{k}_i=(1-\alpha) p^{k}_i+ \alpha$. In this case, to estimate $p^{k}$ we add a post-processing step thresholding to zero all the coordinates of $\widehat \rho^{k}$ obtained from \eqref{def:tv_denoiser} that are smaller or equal to $\alpha$. We test the performance of this method on graphs with various topologies, using the experimental setup described in Section~\ref{sec:num_study}. As shown in Figure~\ref{fig:res_fpr_knn2} for the case of the k-NN graph,
 our procedure can significantly improve the accuracy of the nowcasted estimates $\hat{p}$ over the naive estimates $\hat{p}^{\text{naive}} = y_{\text{observed}}$, particularly as the epidemic size increases. Plots for other graph topologies (small-world, preferential attachment, etc) can be found in Appendix~\ref{appendix:results}.

 \begin{figure}[h!]
     \centering
     \includegraphics[width=\textwidth]{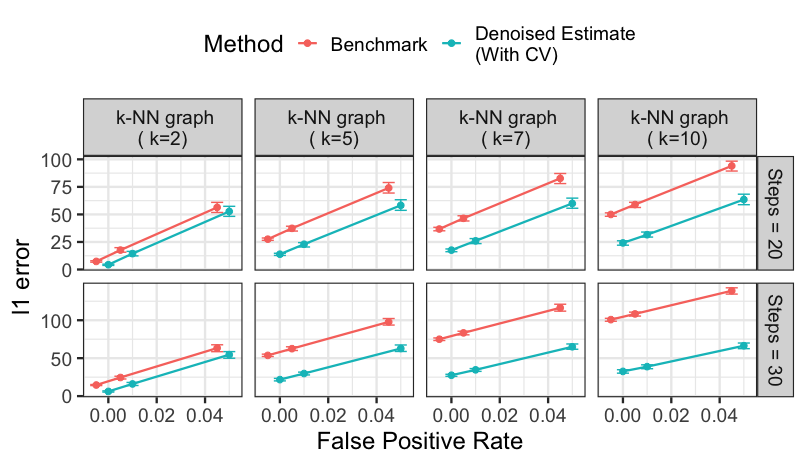}
     \caption{$\ell_1$ error for the k-NN graph, as a function of the noise level (false positive rate) $\alpha \in \{0, 1\%, 5\% \}$, $k_0$ and $\beta$. The healing rate is here fixed at $\gamma=0.1.$. The benchmark values (in red) have been slightly shifted on the x-axis to improve legibility. The results are averaged over 100 independent experiments. Error bars indicate interquartile ranges.}
     \label{fig:res_fpr_knn2}
 \end{figure}
\end{remark}
\begin{remark}[{\bf Estimating the strength of the connections in the network}]
	One of the key quantities that determines the evolution of the epidemic on the network in \eqref{SIS_DT} is the matrix $\Omega=(\omega_{ij})_{(ij)\in [n]\times  [n] }$, which  quantifies the strength of the connections between nodes in the network.  We can estimate $\Omega$ from the observation of a network using the inhomogeneous random graph model assuming that for $i<j$, the elements $A_{ij}$ of the adjacency  matrix $A$ are independent Bernoulli random variables with the success probability $\omega_{ij}\in [0,1]$.  In this case, each entry $\omega_{ij}$ can be interpreted as the probability that the edge $(i,j)$ is present in the graph $\mathcal G$. 
	
	The problem of estimating $\Omega$ from a single observation of the network has been considered in several papers \citep{gao_rate-optimal_2015,klopp_oracle_2017,klopp_optimal_2019,gaucher_maximum_2021} where strong theoretical guarantees have been obtained for this problem.   One of the most popular methods is based on the variational approximation of the likelihood function \citep{celisse_consistency_2012,bickel_asymptotic_2013}.
	Loosely speaking, this method attaches to the latent variables a distribution  with free
	parameters. These parameters are fitted in a way to obtain a distribution close in Kullback-Leibler
	divergence to the true posterior. The
	crux of this approach lies in the realisation that the variational distribution is simpler than the true posterior, which allows to solve
	the optimization problem approximately, see, for example, \citet{gaucher_maximum_2021}.

\end{remark}

\begin{remark} [Learning  parameters $\gamma$ and $\beta$.]
Assume that we know  $p_k, p_{k+1},\dots p_{K}$. We can write the following equation:
\begin{equation*}
\begin{bmatrix}
	p^{k+1}-p^{k}\\
	p^{k+2}-p^{k+1}\\
	\dots\\
	p^{K-1}-p^{K}
\end{bmatrix}
= 	\Phi
\begin{bmatrix}
 \beta\\
 \gamma
\end{bmatrix}
\end{equation*}
where 
\begin{equation*}
\Phi=\begin{bmatrix}
	(I-P^{k})\Omega p^{k}&-p^{k}\\
	\dots\\
	(I-P^{K-1})\Omega p^{K-1}&-p^{K-1}\\
\end{bmatrix}	
\end{equation*} \citet{vrabac_overcoming_2020} propose to estimate the vector of parameters $[\beta,\gamma]^{T}$ using the inverse (or pseudoinverse) of matrix $\Phi$. However, in practice, we do not have access to $p^{j}$ but only to its noisy counterpart $Y_i$, as given by \eqref{model}. In such situations, one can replace $p^{j}$ by its estimator $\widehat{p}^{j}$ given by \eqref{def:tv_denoiser} or \eqref{def:tv_denoiser_missing} and compute the inverse (or pseudoinverse) of matrix $\Phi$ using $\widehat{p}^{j}$. We illustrate the efficiency of this approach in numerical studies in Section \ref{sec:num_study_parameters}.
\end{remark}

\section{Partially observed epidemics}

In \eqref{def:tv_denoiser}, we assume that we observe the state $y_i\in \{0, 1\}$ for any node in the network. However, in many real-world situations, it is challenging to obtain complete information about each individual's disease status. People may not be aware of their infection or may choose not to disclose it. Therefore, methods that can work with partial information become indispensable. In this section, we extend our analysis to the case of partial observations of $Y^{k}$.

Let $\mathfrak{M}=(m_{i})$ be a mask vector. That is, $m_i=1$ if we observe the information about node $i$, and $m_i=0$ if not. We assume that $m_i$ are independent Bernoulli random variables with parameter $\pi_i$: $m_i\sim \mathcal{B}(\pi_i)$ and $\pi_{i}>0$. In the simplest setting, each coefficient is observed with the same probability, i.e., for every $i\in [n]$, $\pi_{i}=\nu$.
In the case of partial observations, we observe two vectors, the mask vector $\mathfrak{M}$ and the vector $\widetilde{Y}= \mathfrak{M} \circ Y$ where $\circ$ denotes the entry-wise (Hadamar) product.
For any $v=(v_i)_{i\in[n]}\in \mathbb{R}^{n}$ we define the weighted $l_2$-norm of $v$ as:
\begin{equation*}
	\Vert v\Vert _{l_2(\pi)}^{2}=\sum_{i}\pi_{i}v^2_{i}.
\end{equation*}
Let 
\begin{equation*}
	\kappa_{\pi}=\left (\underset{i}{\min}\; \pi_i\right )^{-1}
\end{equation*}
 and $\pi^{-1}=\left (\pi_i^{-1}\right )_{i\in [n]}$. It is easy to see that for any $v\in[0,1]^{n}$
\begin{equation}\label{relation_norms}
	\Vert v\Vert_{2}\leq 	\kappa_{\pi}\Vert v\Vert _{l_2(\pi)}.
\end{equation}
We use the partial observations $(\mathfrak{M},\widetilde{Y})$ to construct the TV denoiser $\widehat p_{miss}$ associated to $G$ as any solution of the following minimization problem:
\begin{equation}\label{def:tv_denoiser_missing}
	\widehat{p}_{miss}\in\underset{p\in \mathbb{R}^{n}}
	{\argmin}\left \{\dfrac{1}{n}\Sum m_i\left (\widetilde y_i-p_i\right )^{2}+\lambda \Vert Dp\Vert_1\right \}.
\end{equation}
We prove the following bound on the risk of the TV denoiser estimated on partially observed data:
\begin{theorem}[Risk bound for  TV denoising with partial observations]\label{thm:denoising_missing}

	Assume that $G$ is connected and define the regularization parameter 
	\begin{equation}\label{def: lambda_missing}
		\lambda=\frac{9\sqrt{2}\rho\,\log\left (n\right )}{n}.
	\end{equation}
For $T=\supp \left (Dp^{*}\right )$, the TV denoiser from partial observations, $ \widehat{p}_{miss}$, defined in \eqref{def:tv_denoiser_missing} satisfies 
	\begin{equation} 
		\begin{split}
			{\Vert \widehat p_{miss}-p^{*}\Vert^{2}_{l_2(\pi)}}\leq C\left \{\dfrac{\rho^{2}\, \kappa_{\pi}\vert T\vert }{\kappa^{2}_{T}}+\dfrac{\Vert \pi\Vert_{1}\Vert \pi^{-1}\Vert_{1}}{n^{2}}\right \} \log^{2}(n)
		\end{split}
	\end{equation}
	with probability at least $1-c/n$ where $C,c$ are fixed numerical constants.	
\end{theorem}
The proof of Theorem~\ref{thm:denoising_missing} is provided in Appendix~\ref{appendix:thm:denoising_missing}.  These results therefore allow us to provide theoretical guarantees on the recovery of $p^*$ despite missing observations. In particular, we consider the setting where $\pi_i=v$, in which case, $v$ captures the fraction of observed nodes within the network. In this case, the dependency between the number of missing observations and the quality of the bound becomes clearer: as expected,  the smaller the fraction $v$, the larger the error.

\section{Numerical experiments}\label{sec:num_study}
 In this section, we validate our method through a series of synthetic and semi-synthetic experiments. These controlled experiments allow us to test a variety of data regimes by varying the topology of the contact graph - a dependency highlighted by the presence of the scaling factor $\rho$ in Proposition ~\ref{prp:est_number_infections} - as well as varying the parameters of the epidemic model ($\beta$ and $\gamma$). 

To simulate an epidemic process, we first sample a graph on $n$ nodes, where $n$ is fixed to 1,000 unless otherwise specified.  In the results presented here, we consider a variety of random graph models, including Erd{\"o}s-R{\'e}nyi graphs (where each edge is independently sampled with probability $p$), Stochastic Block Models, Small World Networks (parametrised by their average degree $m$ and a rewiring parameter $p$), Power-Law graphs (parametrised by a parameter, $m$, indicating the number of edges to attach from a new node to existing nodes as the graph is generated), as well as $k$-Nearest Neighbour graphs (generated by uniformly sampling 2D coordinates between 0 and 1, and connecting nodes to their $k$-nearest neighbours in Euclidean norm). To further improve the realism of our experiments and consider graphs reflective of real-world contact network characteristics, we also generate an epidemic process on the Berkeley graph \citep{nr}. The latter is a large social friendship network extracted from Facebook consisting of users (nodes) with edges representing friendship ties. This also allows us to test the performance of our algorithm on a larger-size graph, as this social network has indeed a total of 22,900 nodes and 852,419 edges.

For each experiment, we randomly sample a ``patient 0'', and propagate the epidemic on the graph $\mathcal G$ as per Equation \eqref{SIS_Matrix} for $k_0$ steps, using a specified value of the infection rate $\beta$ and healing rate $\gamma$. The probability of infection for each node here is chosen to be identical: $\beta_{i} = \beta$. We select the interaction strength between each node $i, j$ as $w_{ij} = \frac{1}{d_i \vee d_j}$. This models a situation in which high degree nodes (e.g. influencers) have looser connections to their neighbours, while peripheral nodes with fewer connections might have tighter bounds to their neighbours. The rationale for this choice of prior is reinforced by the observation that the rate of epidemic transmission increases over time. Individuals with a high number of connections, denoted by high degrees, are here presumed to have, in average, brief interactions with each contact. In contrast, individuals with fewer connections tend to represent closer-knit groups, such as family units, where extended interactions are more common. Note that Assumption \ref{Ass_SIS_model} imposes additional conditions on the topology of the graph. For instance, in Erd{\"o}s-R{\'e}nyi (ER) graphs, Assumption \ref{Ass_SIS_model} implies  that $n p_{ER} \beta \leq 1$ so that $p_{ER} = O(\frac{1}{n})$ --- in other words, the ER graph cannot be too dense. This is by no means a limitation of our model, but rather, a way of establishing graph topologies in which our sparsity assumption can be verified.

In our experiments, unless specified otherwise, $\gamma$ is fixed to 0.1. This corresponds to an average recovery time of 9 days, assuming that the recovery is distributed as a negative binomial with a probability $\gamma$ of recovery at each time step.  The result of the simulation is an underlying probability $(p^{\star})_{i\in [n]}$ that each node is infected after $k_0$ epidemic diffusion steps. In this setting, the lower the value of $k_0$, the more localized the epidemic. 
Finally, we generate a vector of observations, sampling each entry (i.e. the infection status of each node) as $y_i \sim \text{Bernouilli}(p^{\star}_i)$.  An example of this diffusion process is presented in Figure~\ref{fig:test}.

\begin{figure}[ht]
   \centering
   
   \begin{subfigure}[b]{0.33\linewidth}
      \includegraphics[width=\linewidth]{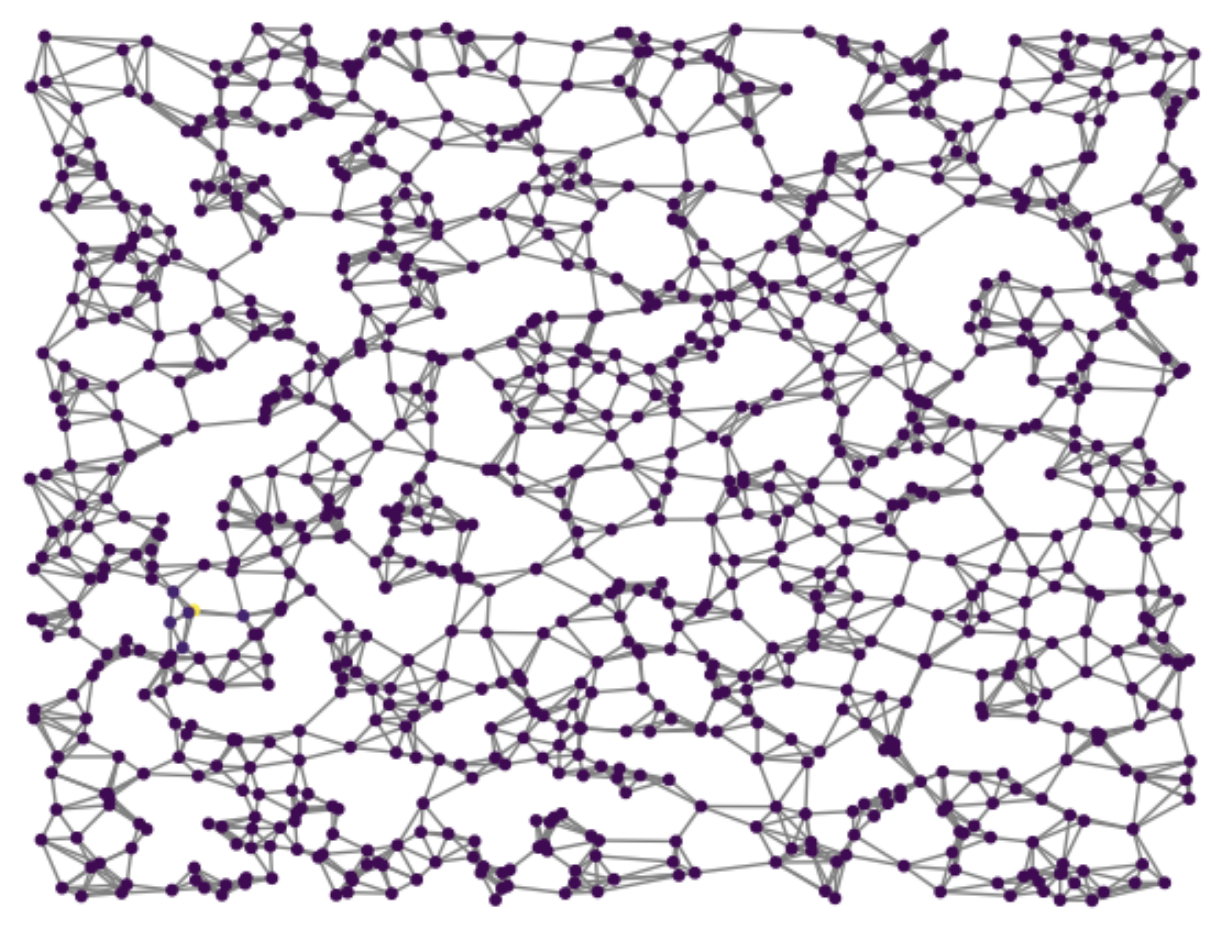}
      \caption{k-Nearest-Neighbour Graph: Time 0}
      \label{fig:sub1}
   \end{subfigure}
   \begin{subfigure}[b]{0.32\linewidth}
      \includegraphics[width=\linewidth]{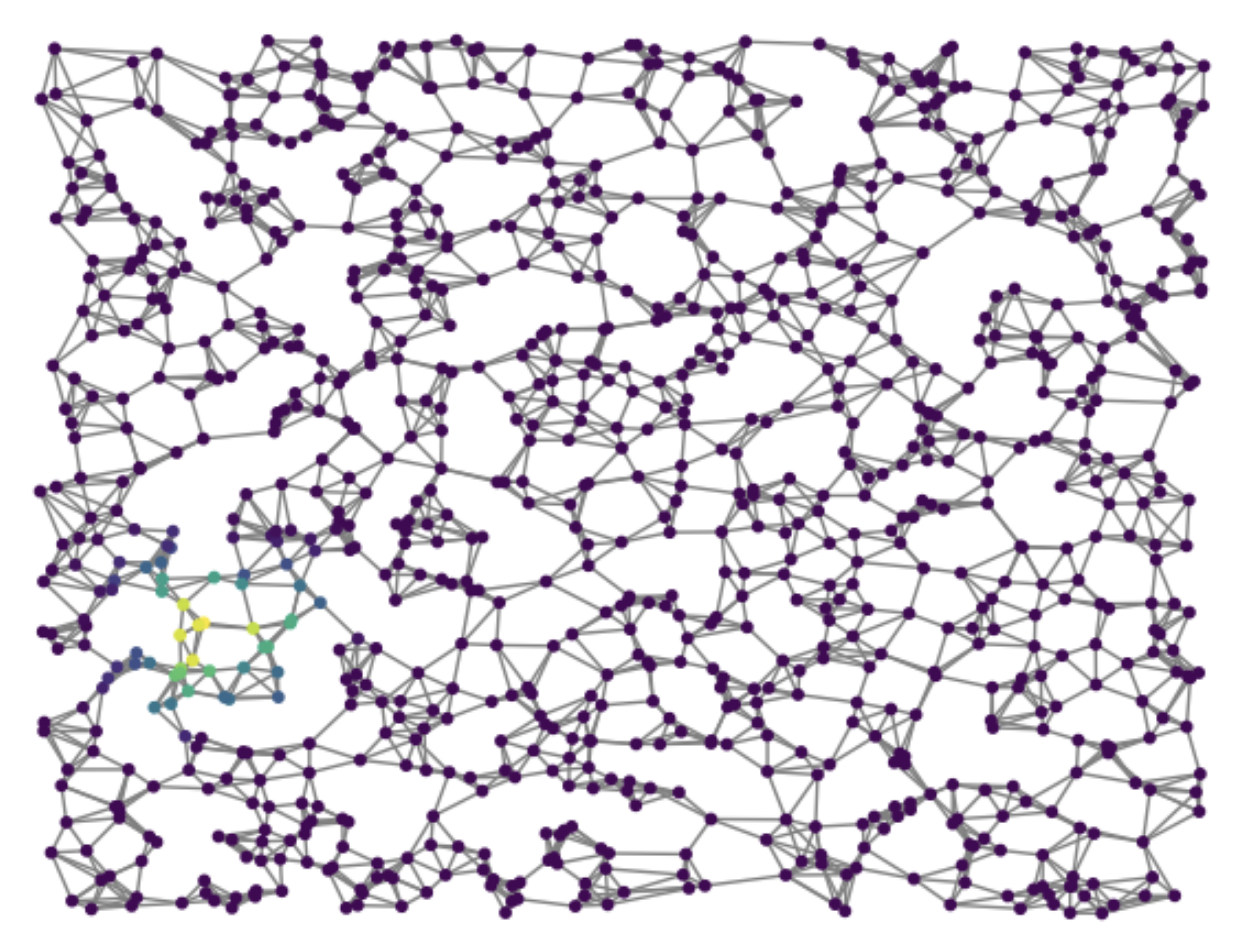}
      \caption{k-Nearest-Neighbour Graph: Time 10}
      \label{fig:sub2}
   \end{subfigure}
      \begin{subfigure}[b]{0.33\linewidth}
      \includegraphics[width=\linewidth]{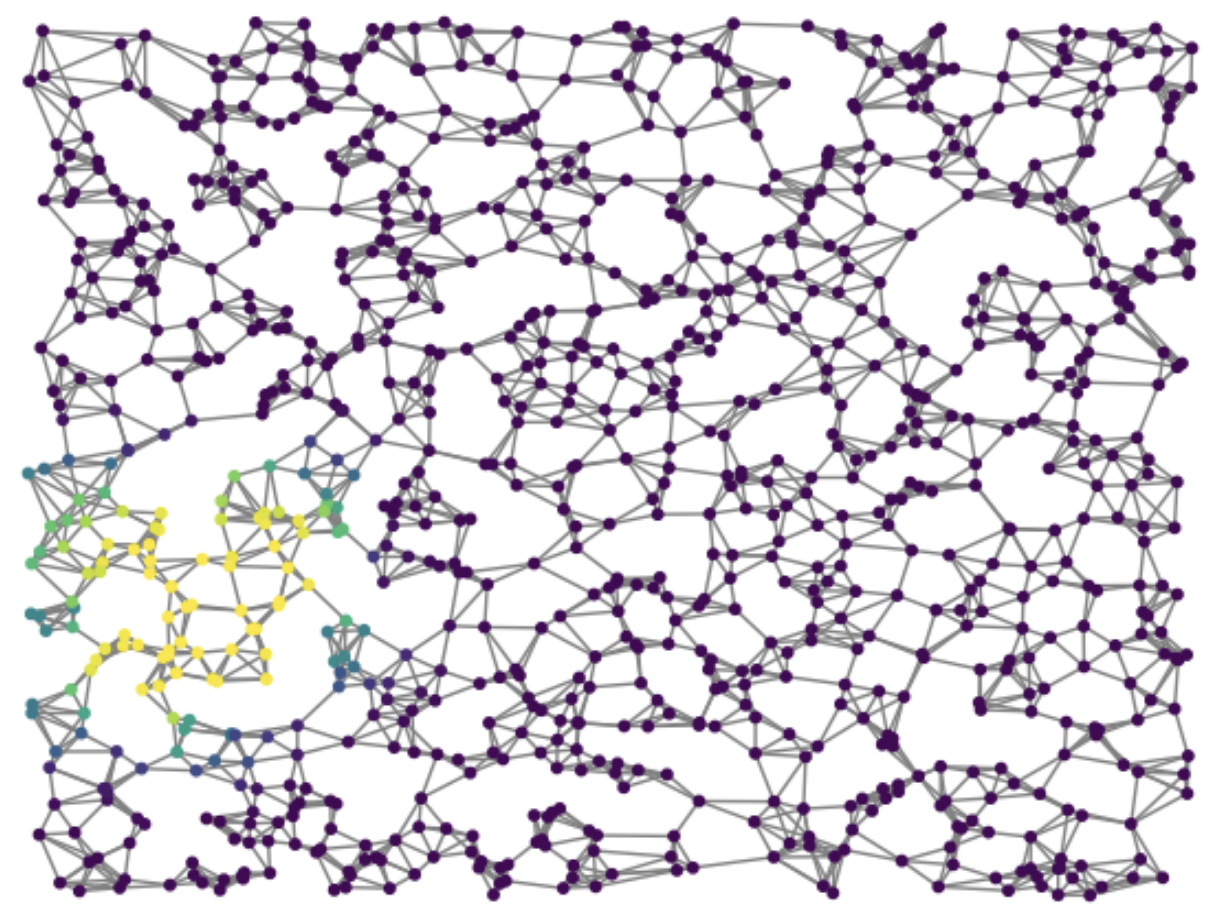}
      \caption{k-Nearest-Neighbour Graph: Time 20}
      \label{fig:sub2}
   \end{subfigure}

      \caption{Visualisation of a SIS epidemic process on a 5-nearest neighbour graph at time 0 (left: one initial infected node), and after $k_0 = 10$ (middle) and $k_0=20$ (right) diffusion steps, for $\beta=0.9$.}
   \label{fig:test}
\end{figure}

\subsection{Algorithm.}
To efficiently solve this problem, we use  the semi smooth-Newton augmented Lagrangian
method (SSNAL) of Sun et al \citep{sun2021convex}.  SSNAL is a scalable and efficient algorithm that was originally designed to solve the following optimization problem:

\begin{equation}\label{eq:SSNAL}
    \min_{X \in \R^{d \times n}} \frac{1}{2} \sum_{i=1}^n \| x_i - a_i\|^2 + \lambda \sum_{i < j } W_{ij} \| x_i -x_j\|_2 
\end{equation}
 where $W_{ij}\in [0,1]$ is a matrix of weights such that $W_{ij}=0$ if nodes $i$ and $j$ are disconnected in the graph, and reflects the proximity between node $i$ and $j$ otherwise. This formulation appears in several settings, particularly in convex clustering problems, where it is used to define a hierarchical clustering of the data.
Note that in our case, since our vector $(p_i)_{i\in [n]}$ is of dimension 1, so $\| p_i -p_j\|_2 = \sqrt{ (p_i - p_j)^2 } = | p_i - p_j|$, and our problem shares the same objective function as Equation~\eqref{eq:SSNAL}. In SSNAL, the optimisation problem \eqref{eq:SSNAL} is solved using an augmented Lagrangian method on equation~\eqref{eq:SSNAL}, and solving the corresponding sub-problems using semi-smooth Newton conjugate gradient updates.
In their work, \citet{sun2021convex} show that as long as the dimension of the feature vector is small (i.e., in our case, equal to 1), this method has the potential to considerably improve existing solvers. We provide a comparison of this solver against alternative solvers (e.g. ADMM~\citep{boyd2011distributed} and CVX \citep{diamond2016cvxpy}) in Appendix~\ref{appendix:results}. In our experiments, an appropriate value for the parameter $\lambda$ is chosen by cross-validation\footnote{The code for all of our simulations and analyses is available at the following link: \texttt{https://github.com/donnate/epidemics}.}. 

\subsection{Results.}
Figures~\ref{fig:denoise_knn_l1} and~\ref{fig:denoise_knn_l2} show the results for the k-NN graph, as the values of $k$, $k_0$ and $\beta$ vary. In these plots,  the $\ell_1$ and $\ell_2$ errors (Figure~\ref{fig:denoise_knn_l1} and ~\ref{fig:denoise_knn_l2}, respectively) of our estimator are presented as a function of the probability of transmission $\beta$. Each column represents a number of diffusion of steps ($k_0$): as emphasized earlier, the higher the value, the more advanced --- and therefore difficult to denoise, from an inference perspective --- the epidemic spread.  Each row corresponds to a different graph topology: in this case, the higher the value of $k$, the denser the contact network. Table~\ref{tab:berkeley} presents the results for the Berkeley graph, and plots corresponding to other topologies are displayed in Figures~\ref{fig:resothers} and~\ref{fig:resother2}, Appendix~\ref{appendix:results}. Overall, we observe that for graphs where the epidemic has not propagated, the denoising method exhibits only marginal improvement over the naive estimator $\hat{p}^{\text{naive}} = y_{\text{observed}}$: in the k-NN and Power-Law graphs for instance, when the number of steps is small, the average number of infections is low (less than 10--- see {Figure~\ref{fig:knn_infection}}), and the one-bit TV denoising method only improve the raw observations by 7\% ($\|\hat{p}-p^*\|_1 = 2.03$ for the 2-NN denoised estimate vs 2.17 for the raw observations at $k_0=10$). However, as the size of the epidemic increases, the effect of the denoising estimator grows considerably. For the 5-NN graph, with $\beta=0.5$ and $k_0=30$, our denoising leads to an almost 50\% decrease of the $\ell_1$ error for the estimated current state of the epidemic.

\begin{figure}[h!]
    \centering
    \includegraphics[width=\textwidth]{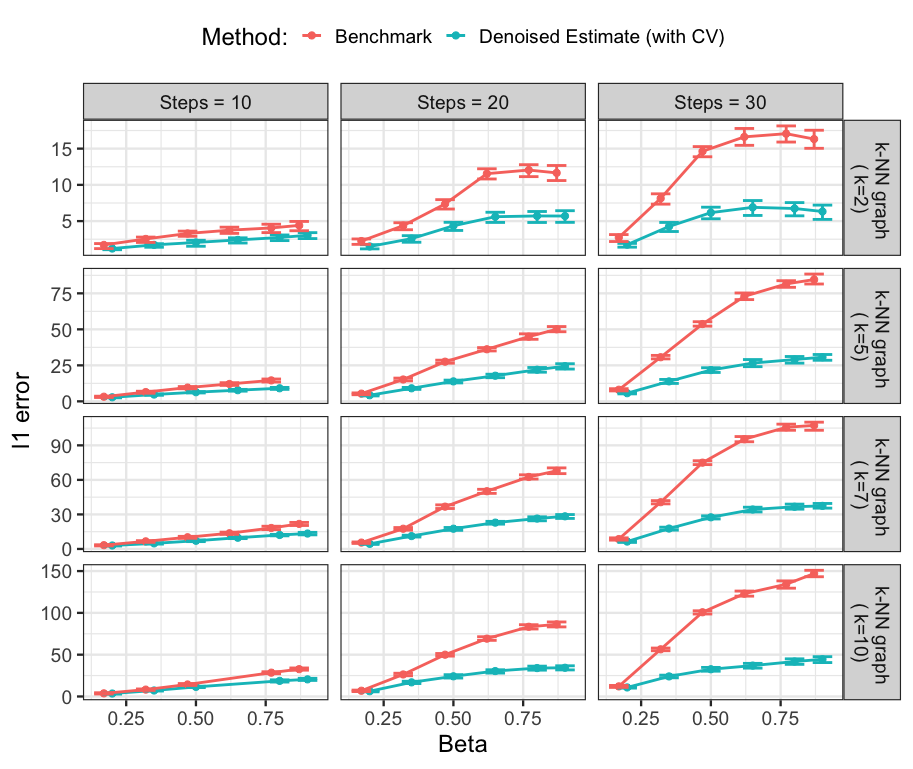}
    \caption{Effect of the denoising (with $\lambda$ chosen through cross validation), as a function of epidemic time (columns) and graph topology (row), measured by the difference in  $\ell_1$-norm between the ground-truth. Points denote the mean $\ell_1 = \|p^* - \hat{p}\|_1$ error over 200 simulations for $k=2$, and 100 simulations for the other topologies. Error bars indicate interquartile ranges. A plot showing the size of the corresponding epidemic can be found in figure~\ref{fig:knn_infection}, {Appendix~\ref{appendix:results}}.}
    \label{fig:denoise_knn_l1}
\end{figure}

\begin{figure}[h!]
    \centering
    \includegraphics[width=\textwidth]{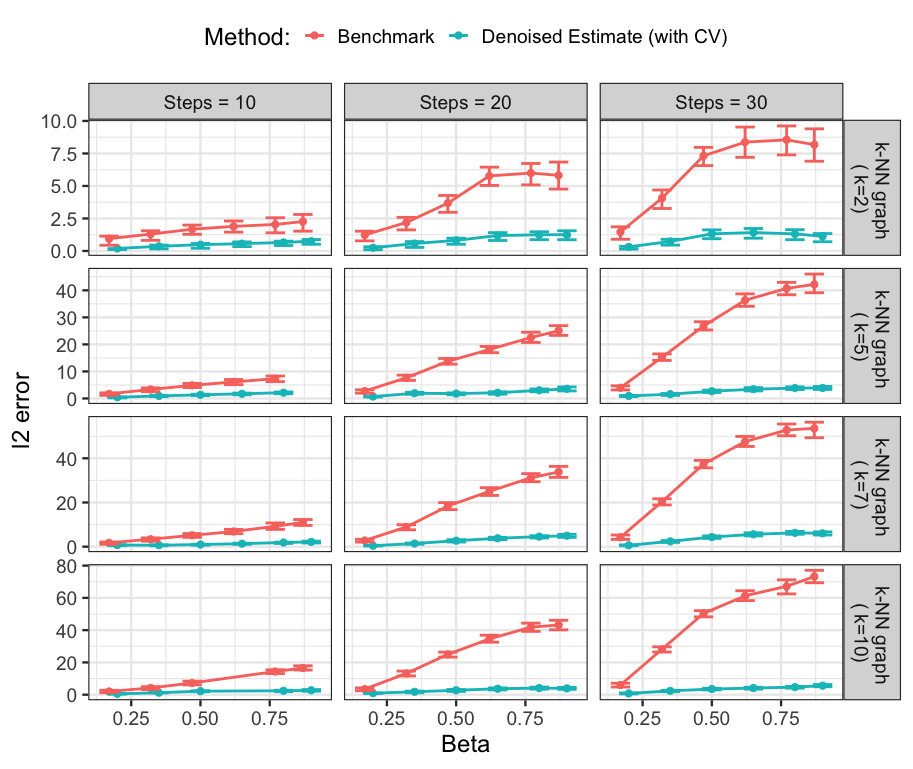}
    \caption{Effect of the denoising (with $\lambda$ chosen through cross validation), as a function of epidemic time (columns) and graph topology (row), measured by the difference in square $\ell_2$-norm between the ground truth and estimated infection parameter $p$. Points denote the mean $\ell_2 = \|p^* - \hat{p}\|_2^2$   error over 200 simulations for $k=2$, and 100 simulations for the other topologies.}
    \label{fig:denoise_knn_l2}
\end{figure}

{Figures~\ref{fig:prediction_knn-2days} shows the $\ell_2$ error in the predicted state of the epidemic in 2 days.} These estimates are obtained by entering the estimated value of $\hat{p}$ in Equation~\eqref{SIS_DT}, and rolling out the updates for two time steps.
 This yields potentially substantial differences in the predictions of the state of the epidemic, as depicted in Figure~\ref{fig:prediction_knn-2days}.

\begin{table}[]
    \centering
\begin{tabular}{|c|c|c|c|}
\hline
  Steps & $\beta$ & $\| \hat{p} - p^*\|_1 $ & $\| \hat{p}^{\text{naive}}- p^*\|_1$ \\
\hline\hline
 5 & 0.7 & 0.00014 & 0.00014\\\hline
5 & 0.9 & 0.00024 & 0.00026\\\hline
10 & 0.3 & 0.00011 & 0.00012\\\hline
 10 & 0.5 & 0.00026 & 0.00032\\\hline
10 & 0.7 & 0.00094 & 0.00124\\\hline
20 & 0.3 & 0.00033 & 0.00042\\\hline
\end{tabular}
    \caption{Accuracy of the estimation (in $\ell_1$ norm) for the Berkeley Graph. Note that in both cases, the accuracy is high, since the epidemic in this case remains very localised.}
    \label{tab:berkeley}
\end{table}

\begin{figure}[h!]
    \centering
    \includegraphics[width=\textwidth]{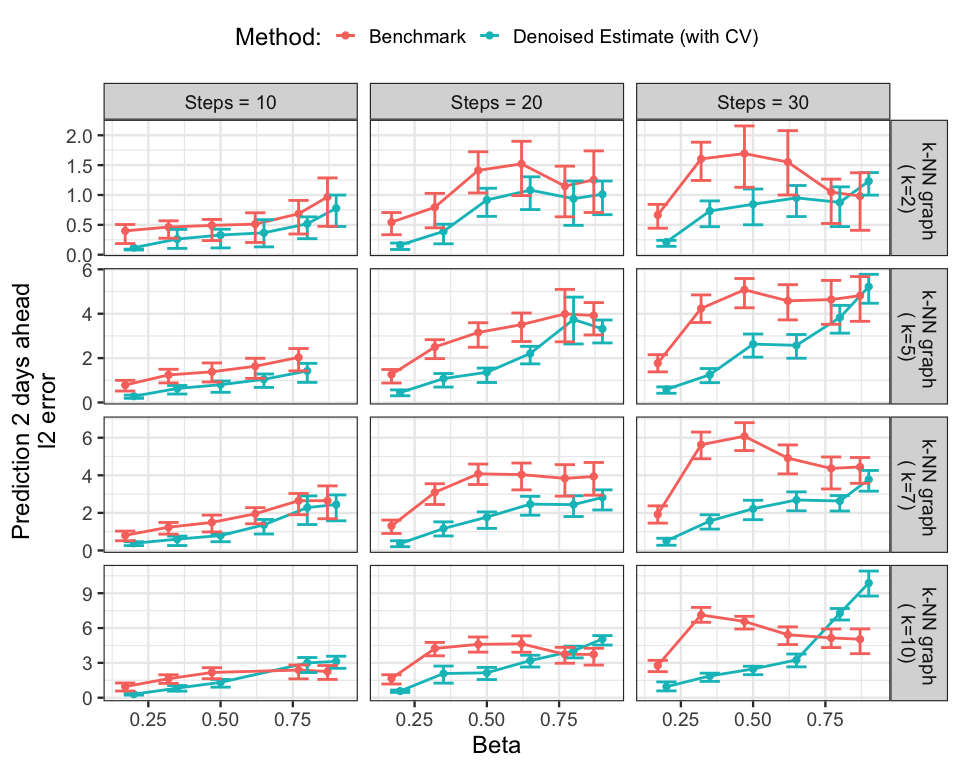}
    \caption{$\ell_2$ error on the state of the epidemic in 2 days, using the predicted $\hat{p}^{(2)}$. Results are compared to the naive estimator that solely uses the observed values: $\hat{p}^{\text{naive}} = y_{\text{observed}}$ to nowcast the epidemic and is then propagated twice as per Equation~\eqref{est_evolution_operator}. The benchmark values (in red) have been slightly shifted on the x-axis to improve legibility. The results are averaged over 100 independent experiments. Error bars indicate interquartile ranges.}
    \label{fig:prediction_knn-2days}
\end{figure}

\subsection{Estimating parameters $\gamma$ and $\beta$.}\label{sec:num_study_parameters}
 To illustrate the efficiency of 1-bit denoiser in the case of estimation of parameters $\gamma$ and $\beta$ we use a kNN-nearest neighbour graph and propagate the epidemic for 20 steps. We vary the values of $\beta$ and $\gamma$, using $k_0=30$. We compare the values of the estimated values of $\beta$ and $\gamma$ as reported by our estimator with those using the same method, but fitting instead the observed values of the state of the epidemic (the naive estimator) . 
The results for the estimates of $\beta$, $\gamma$ and the reproductive number $R_0 = \frac{\beta}{\gamma}$ - an important characteristic of the disease - , averaged over 100 simulations, are reported in Table~\ref{tab:param}. In most cases, our TV denoiser improves the parameter estimates, achieving almost perfect reconstruction for $\beta=0.8$ and $m=5$. In Figure~\ref{fig:params}, we evaluate the effect of the number of observed time steps $k_0$, the graph parameter $k$, and the transmission probability $\beta$ on the estimation error. Overall, we observe that, for this method to be accurate, the epidemic must have a size that is neither too big nor too small: if the epidemic is too big (e.g. as it tends to be the case when $\beta=0.9)$, the sparsity assumption ceases to hold, while if the epidemic is too small (e.g. as it tends to be the case when $\beta=0.2)$, there are too few observations to reliably estimate the parameters.

\begin{table}[]
\resizebox{\textwidth}{!}{%
\begin{tabular}{|c|c|c|c|c|c|c|c|}
\hline
 k-NN& $\beta$ &$\gamma$ & $R_0$ &\textbf{Method} & $\hat \beta$ & $\hat \gamma$ & $\hat R_0$ \\ \hline
 &  &  &  & TV denoiser & \textbf{0.78 (0.58, 1) } & \textbf{0.32 (0.23, 0.44)} & 2.4 (2.27, 2.55)\\ \cline{5-8} 
 & \multirow{-2}{*}{0.35} & \multirow{-2}{*}{0.1} & \multirow{-2}{*}{3.5} & Naive & 0.99 (0.72, 1)  & 0.39 (0.27, 0.52) & \textbf{2.47 (1.91, 2.7) }\\ \cline{2-8} 
 &  &  &  & TV denoiser & \textbf{0.72 (0.58, 0.84)} & \textbf{0.17 (0.13, 0.22)} & \textbf{4.1 (3.89, 4.38)} \\ \cline{5-8} 
 & \multirow{-2}{*}{0.5} & \multirow{-2}{*}{0.1} & \multirow{-2}{*}{5} & Naive & 0.85 (0.69, 1) & 0.21 (0.15, 0.27) & 4.11 (3.76, 4.39) \\ \cline{2-8} 
 &  &  &  & TV denoiser & \textbf{0.84 (0.72, 0.98)} & \textbf{0.11 (0.09, 0.14)} &\bf{7.53 (7.07, 8.2)} \\ \cline{5-8} 
\multirow{-8}{*}{5} & \multirow{-2}{*}{0.8} & \multirow{-2}{*}{0.1} & \multirow{-2}{*}{8} & Naive & 0.94 (0.8, 1) & {0.12 (0.09, 0.14)} & 7.38 (6.21, 8.02) \\ \hline
 &  &  &  & TV denoiser &\textbf{ 0.62 (0.48, 0.78)} & \textbf{0.23 (0.17, 0.31)} & 2.64 (2.51, 2.81) \\ \cline{5-8} 
& \multirow{-2}{*}{0.35} & \multirow{-2}{*}{0.1} & \multirow{-2}{*}{3.5} & Naive & 0.78 (0.6, 0.97) & 0.29 (0.21, 0.38) &  \textbf{ 2.67 (2.53, 2.81)}\\ \cline{2-8} 
&  &  &  & TV denoiser & \textbf{0.66 (0.55, 0.76)} & \textbf{0.15 (0.12, 0.19)} & \textbf{4.29 (4.09, 4.6) } \\ \cline{5-8} 
 & \multirow{-2}{*}{0.5} & \multirow{-2}{*}{0.1} & \multirow{-2}{*}{5} & Naive & 0.77 (0.64, 0.9) & 0.18 (0.14, 0.22) &{4.26 (4.08, 4.53))} \\ \cline{2-8} 
  \multirow{-8}{*}{10} &  &  &  & TV denoiser & \textbf{ 0.83 (0.74, 0.93)} & \textbf{0.11 (0.09, 0.13)} & \textbf{7.75 (7.38, 8.11)}\\ \cline{5-8} 
 & \multirow{-2}{*}{0.8} & \multirow{-2}{*}{0.1} & \multirow{-2}{*}{8} & Naive & 0.91 (0.82, 1)  &  0.12 (0.1, 0.14) & 7.62 (7.11, 7.97) \\ \cline{1-8} 
\end{tabular}
}
\caption{Results for the estimation of the parameters $\gamma$ and $\beta$, as well as the reproductive number $R_0 = \frac{\beta}{\gamma}.$ Intervals denotes interquartile ranges. Results averaged over 500 simulations, for $K_0=30$.}\label{tab:param}
\end{table}

\begin{figure}[h!]
    \centering
    \includegraphics[width=\textwidth]{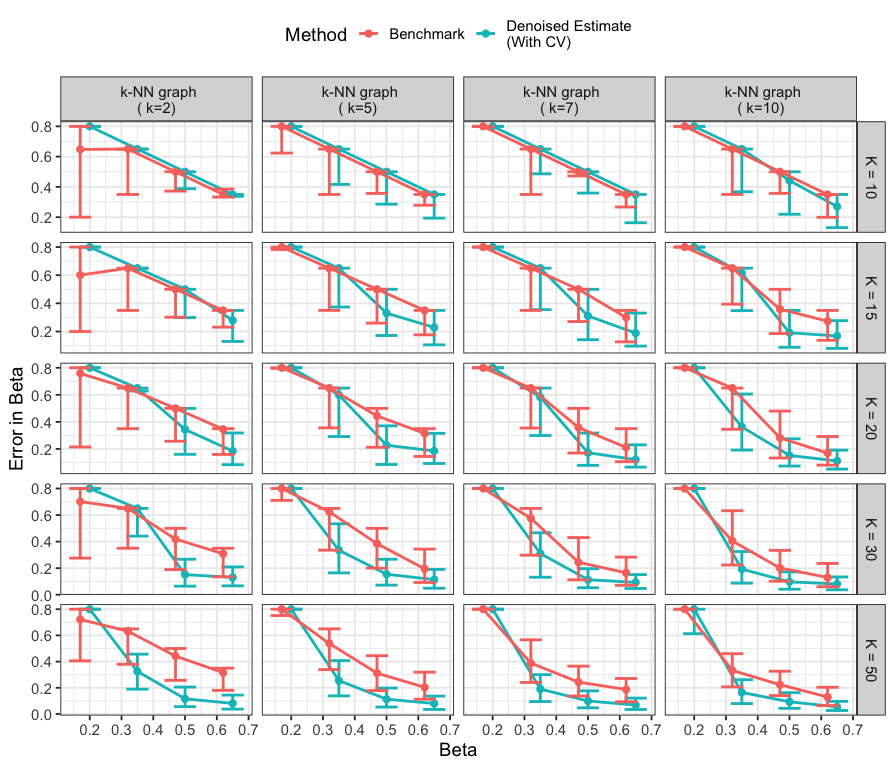}
    \caption{Error in the estimation of $\beta$, as a function of graph density for the kNN graph (parametrized by $k$, represented by each column) and number of observation steps $K$. Results are aggregated over 500 experiments. Points denote the mean error in estimating $\beta$, and error bars indicate the interquartile range. }
    \label{fig:params}
\end{figure}

\subsection{Partially observed epidemic}
 We also study the performance of the 1-bit TV denoiser in the case of partial observations, that is the goodness-of-fit term of our estimator is only applied to nodes whose infectious statuses are observed. We vary the proportion of missing values, and plot the corresponding estimation error. The results are presented in Figure~\ref{fig:missing}. We observe that our estimator can provide substantial improvement over the naive estimator (where all masked entries are filled in with 0) for even large proportions of missing values. This improvement is particularly substantial in cases where the epidemic has progressed enough so that the trivial estimation $p^*_1=0$ is no longer accurate, but not enough to violate our sparsity assumption on the size of the support $\|p^*\|_0.$

\begin{figure}[h!]
    \centering
    \includegraphics[width=\textwidth]{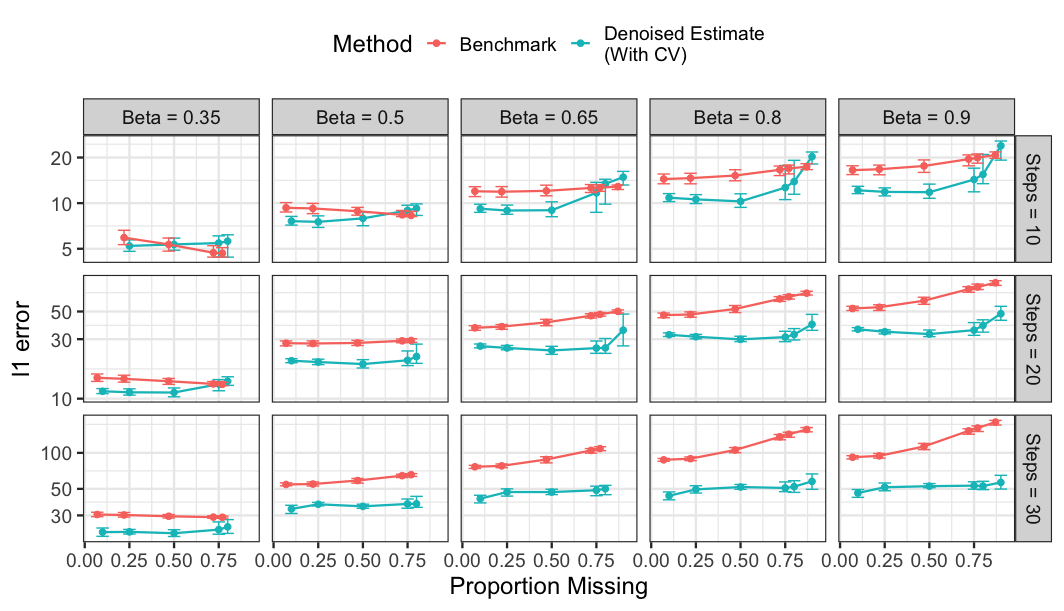}
    \caption{$\ell_1$ error of our one-bit TV denoiser compared to $\hat{p}^{\text{naive}}=y_{\text{observed}}$ in the presence of missing data as a function of the proportion of missing values for the 5-NN graph. Here we vary $k_0 \in \{10, 20, 30\}$ (one value of $k_0$ per row), while columns corresponds to various values of $\beta.$}
    \label{fig:missing}
\end{figure}

\section{Real Data Analysis}
In this section, we propose to analyse two real epidemics datasets using our method.

\subsection{Analysis of the COVID-19 outbreak in California}

To demonstrate the applicability of our method for real-time forecasting (nowcasting), we propose to re-analyse the initial COVID-19 outbreak in California. Our analysis considers the graph formed by the 58 California counties and their corresponding incidence data from January 31, 2020, to April 1, 2020. 
Incidence data are known to suffer from considerable variations throughout the week and underreporting {\citep{graham2020daily,lipsitch2015potential}}, which can significantly skew the estimation of the reproductive number and the predicted epidemic trajectory {\citep{siegenfeld2020models}}. It is, in fact, standard practice to first smooth the raw incidence data through a 7-day rolling average to mitigate irregularities due to weekend reporting (see, for instance, case reporting by the Johns Hopkins Coronavirus Resource Center~\citep{JH}, or by the World Health Organization~\citep{WHO}). However, while the importance of {accounting} for temporal variations seems widely recognized, the extension of this practice to the spatial domain remains unexplored. Yet, due to movement of people between counties, it seems logical to assume that neighbouring counties should expect similar prevalence: any deviations from this pattern could likely be attributed to measurement bias, and should therefore be corrected.
We propose here to correct such biases by employing our method, which we believe leads to a more precise estimation of local disease prevalence. 

We determine the optimal regularization parameter $\lambda$ {by} splitting the data between training and testing. For a given $\lambda$, we train a simple epidemic model on the smoothed data from January 31 to March 14, 2020. This first part of the data is used to estimate the basic reproduction number ($R_0$) and to forecast the following two weeks. Fitting and predictions are performed here using the packages \texttt{R0} and \texttt{earlyR}, which are designed for the estimation and prediction of early-stage epidemics. These packages use the incidence matrix to estimate $R_0$ and predict the epidemic trajectory using simple growth models, which only require as input an estimate of the mean and standard deviation of the serial interval. Following \citet{du2020serial}, we set these values at $\mu=3.95$ and $\sigma=4.75$, respectively.

We then choose the regularization parameter ($\lambda$) that yields the most accurate forecast on the next three days, from March 15 to March 18, 2020. Having chosen an appropriate $\lambda$, we evaluated our model in the period from March 19 to March 31, 2020. The primary reason for this partitioning is that COVID reporting began only in early March for a subset of California counties. However, on March 16, 2020, California implemented a lockdown order, aiming to aggressively reduce $R_0$ and therefore changing the dynamics of the epidemic. However, we anticipate that the incidence immediately after the lockdown reflects the initial dynamics of the epidemic spread, as such changes are not instantaneous given the incubation period of the virus of around 4 to 5 days.

Given the lack of individual-level data and having only access county-level aggregates, we find it easier to model each county as a 'super node' - or equivalently, the entire population as a graph such that each inhabitant is connected to all other inhabitants of the same county, as well as those in neighbouring counties. Letting $n_c$ denote the number of inhabitants in county $c$, this approach thus requires solving the optimisation problem defined as follows:
\begin{equation}
    \argmin_{p \in \R^C} \sum_{c \in [C]}  I_c (1- p_c)^2 + (n_c - I_c)p_c^2 + \lambda \sum_{c, c'} n_c n_c' w_{c,c'} | p_c- p_{c'} |
\end{equation}
where $w_{c,c'}$ is a set of appropriately chosen weights. In the previous equation, $I_c$ denotes the total number of infected inhabitants in county $c$, while $n_c-I_c$ represents the number of susceptible individuals in county $c$.  Here, consistently with our simulation experiments, we pick $w_{c,c'} = \frac{1}{\sum_{c'\sim c} n_{c'}}$ the inverse of the degree of each node. 

The results of this additional preprocessing procedure are presented in Table~\ref{ref:tab_covid}. Overall, for 85\% of the counties on which the epidemic could be estimated (meaning that these counties had at least one COVID case before March 20th), our estimator provides a significant improvement over the raw data. This improvement is particularly striking in well populated and connected areas, such as the San Francisco Bay, where our estimator systematically improves the naive estimator.

\begin{table}[]
\centering
\begin{tabular}{|c|c|c|}
\hline
\textbf{county} & \textbf{\begin{tabular}[c]{@{}c@{}}MSE Prediction\\ (Raw Data)\end{tabular}} & \textbf{\begin{tabular}[c]{@{}c@{}}MSE Prediction\\ (Smoothed Data)\end{tabular}} \\ \hline
Contra Costa & 60 & 27 \\ \hline
Fresno & 0 & 347 \\ \hline
Humboldt & 0 & 0 \\ \hline
Los Angeles & 652 & 425 \\ \hline
Marin & 285 & 17 \\ \hline
Orange & 136 & 73 \\ \hline
Placer & 0 & 1 \\ \hline
Riverside & 32 & 25 \\ \hline
Sacramento & 25 & 18 \\ \hline
San Benito & 1,382 & 297 \\ \hline
San Diego & 2,140 & 716 \\ \hline
San Francisco & 63 & 23 \\ \hline
San Joaquin & 2,157 & 351 \\ \hline
San Mateo & 158 & 114 \\ \hline
Santa Clara & 101 & 78 \\ \hline
Solano & 2 & 0 \\ \hline
Sonoma & 147 & 3 \\ \hline
Stanislaus & 784 & 211 \\ \hline
Tulare & 4,558 & 3 \\ \hline
Ventura & 8,679 & 3,043 \\ \hline
Yolo & 3 & 32 \\ \hline
\end{tabular}
\caption{Comparison of the results in terms of the MSE between predicted epidemic trajectory and realized trajectory for the period from March 20th to March 30th 2023, using the raw data or the spatially smoothed data from our estimator.}\label{ref:tab_covid}
\end{table}

\subsection{The ExFLU Dataset of Aiello et al. \citep{aiello2016design}}

In this second dataset, we reinvestigate the epidemic dataset of \citet{aiello2016design}. This dataset records a flu epidemic outbreak over a college campus in Winter 2013. The study's initial aim was to examine the influence of social interventions {on } the transmission of respiratory infections. A total of 590 university students were enrolled in the study, engaging in weekly surveys regarding symptoms resembling influenza-like illness (ILI) and their social interactions. 
This enables the generation of a weekly contact graph,  which helps in mapping potential epidemic propagation phenomena. 
More specifically here, to evaluate the potential of our method to denoise contact tracing networks, we consider the ego network of each participant declaring experiencing ILI sympoms induced by the contact graph spanning the week of, and the week  after their symptom onset. We define as a contagion event any symptomatic or asymptomatic individual having contracted an ILI illness based on the day of or after the contact's reported symptom onset. This creates a total of 97 ego networks. We then evaluate the ability of our one-bit denoising approach in correctly estimating the ILI statuses of the nodes within this contact network the following week (see Figure \ref{fig:aiello}).

To select the appropriate regularisation parameter, we split the ego networks into two sets, and select lambda on the first, and evaluate on the second. We randomize this procedure and average the results over $100$ iterations. As in the previous example, we quantify the error using the $l_2$ norm. We obtain, in average, an $l_2$ error of $0.0482$, compared to $0.0493$ when setting the regularisation parameter to $0$. Despite the noisy nature of this dataset, we are therefore able to improve upon the naive estimator $(\lambda=0).$ This suggests that our TV denoiser might be effective for epidemic nowcasting using contact tracing networks.

\begin{figure}[h!]
    \centering
    \includegraphics[width=\textwidth]{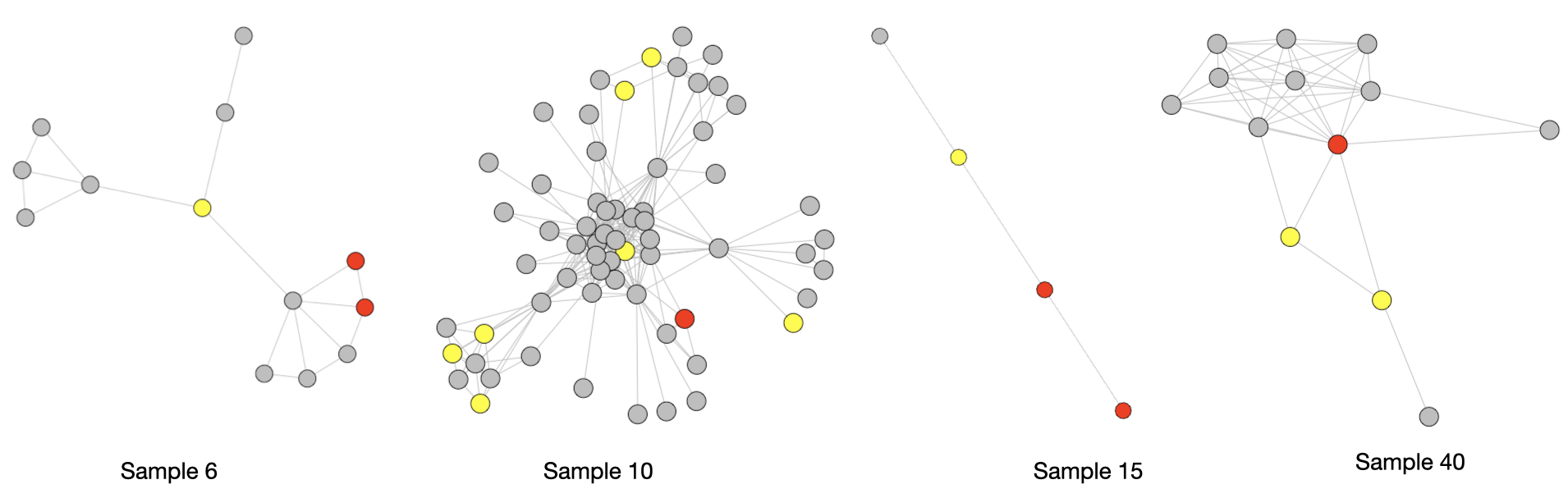}
    \caption{Examples of contact networks induced by taking the 2-hop induced subgraph around participants exhibiting ILI-like symptoms. The contact network here represents all reported contacts between participants on the week of and the week following the symptom onset of the egonet source node. Colours denote different illness statuses. Grey: healthy participant. Yellow: participant sick with ILI symptoms on the week prior to the source node's symptom onset. Red: participant healthy before the source node's symptom onset and developing ILI symptoms in the week following the source node's symptom onset. }
    \label{fig:aiello}
\end{figure}

\section{Related Works}\label{sec:related_works}
In this section we provide a more in-depth overview of related works.

\xhdr{Total Variation Denoising} Total Variation (TV) denoising  (also known as Fused Lasso) has experienced a considerable practical success in the area of image denoising. First introduced in \citet{rudin_nonlinear_1992}, it benefits from the strong theoretical guarantee in the case of Gaussian noise. We refer the readers to the works of \citet{mammen_locally_1997,dalalyan_prediction_2017} for chain graphs and \citet{wang_trend_2016,hutter_optimal_2016, padilla_dfs_2018} for more general graphs.
Earlier works on TV denoising  \citep{mammen_locally_1997} marked an important step by obtaining the first  statistical guarantees.
Another notable development in understanding the behaviour of the TV denoiser was presented in \citet{needell_near-optimal_2013}. In their work, the authors focused on cases where the noise possesses a small $l_2$ norm but is otherwise arbitrary. This framework is commonly encountered in the literature on noisy compressed sensing. An analysis of the statistical
performance of total variation image denoising was performed in \citet{wang_trend_2016} for sub-Gaussian noise and for Gaussian white noise model in \citet{hutter_optimal_2016}  where, in particular, optimal fast rates for the 2D grid are obtained. By contrast, in this work, we built on the rich literature on total variation denoising and provide an analysis of its statistical guarantees in the case of binary observations (see Section \ref{Sec:signal_denoising}). \\

\xhdr{Epidemic spreads on networks} Works on epidemic spread on networks and, more generally, on graph dynamical systems, can be found in various areas. 
Common types of network epidemic models include SIR  and SEIR models, Agent-Based models, Percolation models and etc. Network-based stochastic models, such as the Susceptible-Infectious-Recovered (SIR) and Susceptible-Exposed-Infectious-Recovered (SEIR) models, are commonly used to simulate epidemic dynamics. Transition rates between compartments are estimated based on network structure and contact patterns \citep{newman_spread_2002, kenah_network-based_2007}. Researchers have extended these models to incorporate heterogeneous mixing patterns and complex network structures \citep{miller_incorporating_2013}.
Agent-based models (ABMs)  have gained popularity for simulating realistic human behavior within network contexts. For example, \citet{balcan_multiscale_2009} used ABMs to investigate the impact of human mobility on epidemic spread, considering factors like travel patterns and contact behaviour. These models provide insights into the effectiveness of public health interventions and behavioral changes during outbreaks.

Stochastic network models, such as the temporal exponential random graph model (TERGM), capture dynamic changes in network structure over time. They are employed to study how changes in contact patterns impact epidemic dynamics (see, for example, \citet{volz_susceptibleinfectedrecovered_2007}). These models provide a statistical framework for exploring network evolution during epidemics.

Discrete-time Markov Processes (DTMP) and Continuous-time Markov Processes (CTMP) are commonly used in network epidemic models to represent the dynamics of disease spread or information diffusion within a population \citep{gomez_discrete-time_2010, zhang_diffusion_2014}. In DTMP, time is divided into discrete intervals or time steps. For example, in a SIR (Susceptible-Infectious-Recovered) model, susceptible individuals may become infected with a certain probability during each time step. In Continuous-time Markov Processes (CTMP), time is treated as a continuous variable, allowing for modeling events that occur at any moment in time. State transitions can occur at any time, not just at discrete time steps.
Events such as contacts, infections, recoveries, or interventions are represented as stochastic processes. The rates at which these events occur depend on factors such as contact frequencies, transmission probabilities, and recovery rates.

Bayesian inference methods have gained popularity for parameter estimation in epidemic models. Markov Chain Monte Carlo (MCMC) algorithms, such as the Metropolis-Hastings algorithm, are used to sample from posterior distributions of model parameters \citep{jewell_bayesian_2009,hu_bayesian_2017}. Bayesian frameworks allow, in particular for uncertainty quantification and model selection. However, to the best of our knowledge, no statistical guarantees have yet been provided for the recovery of the parameters.

Sampling and contact tracing data are essential for estimating epidemic parameters. \citet{lloyd-smith_superspreading_2005} used network-based contact tracing to reconstruct transmission chains in emerging infectious diseases. This approach provides insights into superspreading events and helps identify high-risk nodes. Statistical techniques for inferring contact networks from observed epidemic data have been also developed. For instance, \citet{rocha_simulated_2011} introduced a likelihood-based method for reconstructing contact networks using observed infection times. Such methods are crucial for understanding the underlying network structure in real-world epidemics.

The study of infectious disease spread in the context of partially observed epidemics is essential in real-world situations where complete information about individuals' health status is often unavailable or uncertain. Research in this area has addressed various aspects of the problem, including modeling, inference, and control strategies (for a review of the latter, see for instance \citet{britton_2016}).
Several models have been developed that incorporate partial information. For example, the Susceptible-Infectious-Unknown (SIU) model extends the classic Susceptible-Infectious-Recovered (SIR) model to account for individuals with unknown health status.

Within the broader context of partially-observed Markov processes, several techniques have been incorporated into the R package POMP \citep{king_2016}. Notably, one of these methods is maximum iterated filtering (MIF) \citep{ionides_2011}. Other approaches include Gaussian approximation of
the epidemic density-dependent Markovian jump process \citep{narci_2021}.
Additionally, likelihood-free approaches such as approximate Bayesian computation based on sequential Monte Carlo \citep{toni_2009,Sisson_2007} and particle Markov chain Monte Carlo \citep{andrieu_2010} have emerged. These approaches aim to infer the hidden states of individuals (infected or healthy) based on observed data, such as reported cases or symptom data.


\section*{Acknowledgments}
Claire Donnat gratefully acknowledges {support} from NSF Award Number RI:2238616, as well as the resources provided by the University of Chicago’s Research Computing Center.
The work of Olga Klopp was funded by CY Initiative (grant “Investissements d’Avenir” ANR-16-IDEX-0008) and Labex MME-DII (ANR11-LBX-0023-01). This work was done  while O. Klopp  and C. Donnat were visiting the Simons Institute for the Theory of Computing.  The work of N. Verzelen has been partially supported by ANR-21-CE23-0035 (ASCAI).

\section*{Supplementary material}
\label{SM}
The Supplementary Material includes the proofs of the theorems and proposition established in this paper, as well as further results stemming from our synthetic experiments.

\appendix
\appendixone
\section{Proof of Theorem \ref{thm:denoising}}\label{appendix:thm:denoising}
The proof of  Theorem \ref{thm:denoising} is close in spirit to the proof of the Theorem 2 in \citep{hutter_optimal_2016} the main difference being the control of the noise term. 

Applying the chain rule, we have that  the subdifferential of the $l_1$ term is given by
\[\partial \Vert Dp\Vert_{1}=D^{T}\sign (Dp)\]
where 
\begin{equation*}
	\sign(x)_i=\begin{cases}
		1& \text{if } x_i>0,\\
		[-1,1]& \text{if } x_i=0,\\
		-1& \text{if} x_i<0.\\
	\end{cases}
\end{equation*}
Using first order optimality condition for the convex problem \eqref{def:tv_denoiser}, we derive that there exists $z\in \sign (D\widehat p)$ such that  for any $\bar p\in \mathbb{R}^n$,    we have
\begin{equation}\label{eq:1}
	\dfrac{2}{n}\langle \bar p, Y-\widehat{p}\rangle= \lambda \langle \bar{p},D^{T}z\rangle=\lambda \langle D\bar{p},z\rangle\enspace .
\end{equation}
This  implies that 
\[\dfrac{2}{n}\langle \widehat p, Y-\widehat{p}\rangle=\lambda \Vert D\widehat{p}\Vert_{1}\quad \text{and}
\]
\[\dfrac{2}{n}\langle  p^{*}, Y-\widehat{p}\rangle\leq \lambda \Vert D{p^{*}}\Vert_{1}.
\]
Subtracting the two terms above, we get 
\[\dfrac{2}{n}\Vert p^{*}-\widehat{p}\Vert^{2}_{2}\leq \dfrac{2}{n}\langle\xi,\widehat{p}-p^{*}\rangle+\lambda\Vert Dp^{*}\Vert_{1}-\lambda\Vert D\widehat p\Vert_{1}.\]
Next, we control the error term. Let $\Pi$ denote the projection on $\ker D$, the kernel of $D$. Then, by the definition of $D^{\dagger}$, $D^{\dagger}D=I-\Pi$ is the projection on $(\ker D)^{\perp}$. We have that $\ker D= \ker (D^{T}D)=\ker L$. We assume that $G$ is connected, which implies that zero is an eigenvalue of $L$ of multiplicity one and $\dim\ker D=1$. We can write
	\begin{align}\label{eq:2}
		\langle \xi  , \widehat{p}- p^{*}\rangle&= \langle \Pi\xi  , \widehat{p}- p^{*}\rangle + \langle (I-\Pi)\xi  , \widehat{p}- p^{*}\rangle\nonumber\\
		&= \langle \Pi\xi  , \Pi(\widehat{p}- p^{*})\rangle + \langle D^{\dagger}D\,\xi  , \widehat{p}- p^{*}\rangle
	\nonumber	\\
		&= \langle \Pi\xi  , \Pi(\widehat{p}- p^{*})\rangle + \langle \xi  , D^{\dagger}D\left (\widehat{p}- p^{*}\right )\rangle
	\nonumber	\\
		&= \langle \Pi\xi  , \Pi(\widehat{p}- p^{*})\rangle + \langle \left (D^{\dagger}\right )^{T} \xi  , D\left (\widehat{p}- p^{*}\right )\rangle \enspace . 
	\end{align}
To bound the right hand side in \eqref{eq:2}, we first use the H\"older inequality to get
\begin{align*}
	\langle \Pi\xi  , \Pi(\widehat{p}- p^{*})\rangle\leq \Vert \Pi\xi\Vert_{2}\Vert \Pi(\widehat{p}- p^{*})\Vert_{2}\enspace .
\end{align*}
Note that $L\,\mathbf{1}_{n}=\mathbf{0}_{n}$ and, for any vector $v\in \mathbb{R}^{n}$, $\Pi v=\frac{1}{n}\langle v,\mathbf{1}_{n}\rangle \mathbf{1}_{n}$ which implies 
$$\Vert \Pi\xi\Vert_{2}= \tfrac{1}{\sqrt{n}}\vert \langle \xi,\mathbf{1}_{n}\rangle\vert.$$
Now, using Hoeffding's inequality 
 we get that, with probability at least $1-\delta/2$, 
\begin{equation}\label{eq:3}
\vert \langle \xi,\mathbf{1}_{n}\rangle\vert=\left \vert \sum_{i=1}^{n}\xi_i\right \vert\leq \sqrt{2\Vert p^{*}\Vert_{0}\log\left (\tfrac{4}{\delta}\right )}\,\quad \text{and}\quad \Vert \Pi\xi\Vert_{2}\leq \sqrt{\frac{2\Vert p^{*}\Vert_{0}\,\log\left (\tfrac{4}{\delta}\right )}{n}}\enspace .
\end{equation}
On the other hand, we get
\begin{equation}\label{eq:4}
	\Vert \Pi(\widehat{p}- p^{*})\Vert_{2}=\tfrac{1}{\sqrt{n}}\vert \langle \widehat{p}- p^{*},\mathbf{1}_{n}\rangle\vert\leq \tfrac{1}{\sqrt{n}}\Vert  \widehat{p}- p^{*}\Vert_{1}
\end{equation} 
and (\ref{eq:3}--\ref{eq:4}) imply
\begin{equation}\label{eq:5}
	\langle \Pi\xi  , \Pi(\widehat{p}- p^{*})\rangle\leq \frac{\sqrt{2\Vert p^{*}\Vert_{0}\,\log\left (\tfrac{4}{\delta}\right )}}{n}\Vert  \widehat{p}- p^{*}\Vert_{1}.
\end{equation}
For the second term in \eqref{eq:2} we have that, with probability at least $1-\delta/2$,
\begin{align}\label{eq:6}
\langle \left (D^{\dagger}\right )^{T} \xi  , D\left (\widehat{p}- p^{*}\right )\rangle &\leq \Vert \left (D^{\dagger}\right )^{T} \xi\Vert_{\infty}	\Vert D\left (\widehat{p}- p^{*}\right )\Vert_{1}\nonumber\\
&\leq \sqrt{2}\rho\log\left (\frac{4n^{2}}{\delta}\right )	\Vert D\left (\widehat{p}- p^{*}\right )\Vert_{1}\enspace ,
\end{align}
where the last inequality follows from Lemma \ref{lem:noise_term}. Putting together \eqref{eq:5} and \eqref{eq:6} we get that, with probability at least $1-\delta$,
\begin{align*}
	\langle \xi  , \widehat{p}- p^{*}\rangle\leq \sqrt{2}\rho\log\left (\frac{4n^{2}}{\delta}\right )	\Vert D\left (\widehat{p}- p^{*}\right )\Vert_{1}+ \frac{\sqrt{2\|p^*\|_0\,}}{n}\log\left (\frac{4}{\delta}\right )\Vert  \widehat{p}- p^{*}\Vert_{1}\enspace , 
	\end{align*}
	since $\|p^*\|_1\leq \|p^*\|_0$. This leads us to
\begin{align*}
\Vert p^{*}-\widehat{p}\Vert^{2}_{2}\leq \sqrt{2}\rho\log\left (\frac{4n^{2}}{\delta}\right )	\Vert D\left (\widehat{p}- p^{*}\right )\Vert_{1}&+n\lambda\Vert Dp^{*}\Vert_{1}-n\lambda\Vert D\widehat p\Vert_{1}\nonumber\\
&+ \frac{\sqrt{2\|p^*\|_0\,\log\left (\tfrac{4}{\delta}\right )}}{n}\Vert  \widehat{p}- p^{*}\Vert_{1}\enspace .
\end{align*}
Using $	n\lambda=\sqrt{2}\rho\log\left (\frac{4n^{2}}{\delta}\right )$ and the triangle inequality we get
\begin{align*}
	\Vert p^{*}-\widehat{p}\Vert^{2}_{2}\leq 2n\lambda\left (	\Vert D\left (\widehat{p}- p^{*}\right )_{T}\Vert_{1}+\Vert \left (Dp^{*}\right )_{T^{c}}\Vert_{1}\right )+ \frac{2\|p^*\|_0\log\left (\tfrac{4}{\delta}\right )}{n}+\frac{1}{4}\Vert  \widehat{p}- p^{*}\Vert^{2}_{2}\enspace .
\end{align*}
Now, applying Definition \ref{def:compatibility factor} of compatibility factor we can write
\begin{align*}
	\frac{1}{2}\Vert p^{*}-\widehat{p}\Vert^{2}_{2}\leq \dfrac{8\rho^{2}\vert T\vert}{\kappa^{2}_{T}}\log^{2}\left (\frac{4n^{2}}{\delta}\right )+2\sqrt{2}\rho\log\left (\frac{4n^{2}}{\delta}\right )\Vert \left (Dp^{*}\right )_{T^{c}}\Vert_{1}+ \frac{2\|p^*\|_0\log\left (\tfrac{4}{\delta}\right )}{n}
\end{align*}
which proves the statement of Theorem \ref{thm:denoising}.

\subsection{Lemma \ref{lem:noise_term}}
\begin{lemma}\label{lem:noise_term}
		Fix $\delta\in[0,1]$. For $i\in [n]$, let $z_i$ be independent centered Bernoulli random variables with parameter $q_i$ and $z=(z_i)_{i\in[n]}$. Then, with probability at least $1-\delta$, 
		\[\Vert \left (D^{\dagger}\right )^{T} z\Vert_{\infty}\leq \sqrt{2}\rho\log\left (\frac{2n^{2}}{\delta}\right ).\]	
\end{lemma}

\begin{proof}
	We start by applying Bernstein's inequality to $\sum_{i=1}^{n}D^{\dagger}_{ij}z_i$. Note that $ \left \vert D^{\dagger}_{ij}z_i\right \vert \leq \rho$ and $\sum_{i=1}^{n}\left (D^{\dagger}_{ij}\right )^{2}\bE \,z^{2}_i\leq \sum_{i=1}^{n}\left (D^{\dagger}_{ij}\right )^{2}q_{i}\leq \rho^{2}$.  Then, the Bernstein's inequality implies that, with probability at least $1-\frac{1}{n^{2}\delta}$, 
\[\left \vert	\sum_{i=1}^{n}D^{\dagger}_{ij}z_i \right \vert\leq \sqrt{2}\rho\log\left (\frac{2n^{2}}{\delta}\right )\ .\]
Now, applying the union bound and using $m\leq n^{2}$ we get the result of Lemma \ref{lem:noise_term}.
\end{proof}

\appendixtwo
\section{Proof of Proposition \ref{prp:risk_l1}}\label{appendix:prp:risk_l1}

 Take $\bar p=(\bar{p}_{i})^{n}_{i=1}$ such that 
\begin{equation*}
	\bar p_{i}=\begin{cases}
		-1& \text{if }\; i\in S^{c},\\
		0& \text{if not. }
	\end{cases}
\end{equation*}
Then, the first order optimality condition \eqref{eq:1} implies  
\begin{equation*}
	\langle \bar p, \left (p^{*}-\widehat{p}\right )_{S}\rangle - 	\langle \bar p, \widehat{p}_{S^{c}}\rangle +	\langle \bar p,\xi\rangle=\frac{n\lambda}{2} \langle D\bar{p},z\rangle.
\end{equation*}
Note that $\xi_{S^{c}}=0$ which implies
\begin{equation*}
	\Vert \widehat{p}_{S^{c}}\Vert_{1}=\frac{n\lambda}{2} \langle D\bar{p},z\rangle\leq \frac{n\lambda}{2}\Vert D{p}^*\Vert_{0}.
\end{equation*}
We have that for any $e=(i,j)\in E$, $(D\bar{p})_{e}=0$ if $(i,j)\in S\times S$ or if $(i,j)\in S^{c}\times S^{c}$ which implies that $\langle D\bar{p},z\rangle\leq \Vert D{p}^*\Vert_{0}$. Then we compute
\begin{align*}
	\Vert p^{*}-\widehat{p}\Vert_{1}=\Vert \left (p^{*}-\widehat{p}\right )_{S}\Vert_{1}+\Vert \widehat{p}_{S^{c}}\Vert_{1}\leq \frac{n\lambda}{2} \Vert Dp^*\Vert_0 + \sqrt{s}\Vert p^{*}-\widehat{p}\Vert_{2}.
\end{align*}
Combining this bound with Theorem \ref{thm:denoising} we get the statement of Proposition \ref{prp:risk_l1}.

\appendixthree
\section{Proof of Theorem \ref{thm:denoising_missing}}\label{appendix:thm:denoising_missing}

For simplicity,  we drop the $miss$ index, that is we write $\widehat{p}_{miss}=\widehat{p}$. By first order optimality condition for the convex problem \eqref{def:tv_denoiser_missing}, there exists 
$z\in \sign (D\widehat p)$ such that for any $\bar p\in \mathbb{R}$,  we have
\begin{equation}\label{eq:1_missing}
	\tfrac{2}{n}\left \langle \bar p, \mathfrak{M}\circ\left ( \widetilde{Y}-\widehat{p}\right )
	\right \rangle= \lambda \langle \bar{p},D^{T}z\rangle=\lambda \langle D\bar{p},z\rangle\ , 
\end{equation}
which implies 
\[\tfrac{2}{n}\left \langle \widehat p, \mathfrak{M}\circ\left ( \widetilde{Y}-\widehat{p}\right )\right \rangle=\lambda \Vert D\widehat{p}\Vert_{1}\quad \text{and}
\]
\[\tfrac{2}{n}\left \langle p^{*}, \mathfrak{M}\circ\left ( \widetilde{Y}-\widehat{p}\right )\right \rangle\leq \lambda \Vert D{p^{*}}\Vert_{1}.
\]
Subtracting the above two, we obtain 
\[\tfrac{2}{n}\left \Vert \mathfrak{M}\circ \left (p^{*}-\widehat{p}\right )\right \Vert^{2}_{2}\leq \tfrac{2}{n}\langle \mathfrak{M}\circ \xi,\widehat{p}-p^{*}\rangle+\lambda\Vert Dp^{*}\Vert_{1}-\lambda\Vert D\widehat p\Vert_{1}.\]
The control of the error term is similar to the case with complete observations and we only sketch it.
 Let $\eta=\mathfrak{M}\circ \xi=(\eta_{i})_{i\in [n]}$. Note that the $\eta_{i}$'s are independent Bernoulli random variables with parameter $\pi_{i}p_i^{*}$. Then, following similar argument as in the proof of Theorem \ref{thm:denoising} and using Bernstein's inequality, we get  with probability at least $1-\delta/2$, 
\begin{align}\label{eq:2_missing}
	\langle \Pi\eta  , \Pi\left (\widehat{p}- p^{*}\right )\rangle&\leq \Vert \Pi\eta\Vert_{2}\left \Vert  \Pi\left (\widehat{p}- p^{*}\right )\right \Vert_{2}\\
	&\leq \frac{\sqrt{2\Vert \pi\circ p^{*}\Vert_{1}\,}\log\left (\tfrac{4}{\delta}\right )}{n}\;\Vert  \widehat{p}- p^{*}\Vert_{1}\nonumber\\
	&\leq \frac{\sqrt{2\Vert \pi\circ p^{*}\Vert_{1}\,\Vert \pi^{-1}\Vert_{1}}\log\left (\tfrac{4}{\delta}\right )}{n}\;\Vert  \widehat{p}- p^{*}\Vert_{l_2(\pi)}\nonumber
\end{align}

%
Similarly,  we have that with probability $1-\delta/2$
\begin{align}\label{eq:3_missing}
	\left\langle \left (D^{\dagger}\right )^{T} \eta  , D\left (\widehat{p}- p^{*}\right )\right\rangle
	 &\leq \Vert \left (D^{\dagger}\right )^{T} \eta\Vert_{\infty}	\Vert D\left (\widehat{p}- p^{*}\right )\Vert_{1}\nonumber\\
	&\leq \sqrt{2}\rho\log\left (\frac{4n^{2}}{\delta}\right )	\Vert D\left (\widehat{p}- p^{*}\right )\Vert_{1}
\end{align}
where the last inequality follows from Lemma \ref{lem:noise_term}. Putting together \eqref{eq:2_missing} and \eqref{eq:3_missing} we get that, with probability at least $1-\delta$,
\begin{align*}
	\langle \mathfrak{M}\circ \xi, \widehat{p}- p^{*}\rangle\leq & \sqrt{2}\rho\log\left (\frac{4n^{2}}{\delta}\right )	\Vert D\left (\widehat{p}- p^{*}\right )\Vert_{1}\\
	&\hskip 0.5 cm +
	\frac{\sqrt{2\Vert \pi\circ p^{*}\Vert_{1}\,\Vert \pi^{-1}\Vert_{1}}\log\left (\tfrac{4}{\delta}\right )}{n}\;\Vert  \widehat{p}- p^{*}\Vert_{l_2(\pi)}
\end{align*}
and 
\begin{align*}
	\left \Vert \mathfrak{M}\circ \left (p^{*}-\widehat{p}\right )\right \Vert^{2}_{2}\leq \sqrt{2}\rho&\log\left (\frac{4n^{2}}{\delta}\right )	\Vert D\left (\widehat{p}- p^{*}\right )\Vert_{1}+\dfrac{n\lambda}{2}\Vert Dp^{*}\Vert_{1}-\dfrac{n\lambda}{2}\Vert D\widehat p\Vert_{1}\nonumber\\
	&+ \frac{\sqrt{2\Vert \pi\circ p^{*}\Vert_{1}\,\Vert \pi^{-1}\Vert_{1}}\log\left (\tfrac{4}{\delta}\right )}{n}\;\Vert  \widehat{p}- p^{*}\Vert_{l_2(\pi)}\enspace .
\end{align*}
For $T=\supp \left (Dp^{*}\right )$ and $\delta=4/n$, using $	n\lambda=9\sqrt{2}\rho\log\left (n\right )$,  we can write
\begin{align}\label{eq:4_missing}
	\left \Vert \mathfrak{M}\circ \left (p^{*}-\widehat{p}\right )\right \Vert^{2}_{2}+	\frac{n\lambda}{6}\left \Vert D\,\widehat p_{T^{C}}\right \Vert_{1}\leq&  \frac{5n\lambda}{6}\left \Vert D(p^{*}-\widehat p)_{T}\right \Vert_{1}\\
	& \hskip 0.5 cm+ 
	\frac{\sqrt{2\Vert \pi\circ p^{*}\Vert_{1}\,\Vert \pi^{-1}\Vert_{1}}\log\left (n\right )}{n}\;\Vert  \widehat{p}- p^{*}\Vert_{l_2(\pi)}\nonumber
\end{align}
and 
\begin{align}\label{eq:5_missing}
	\left \Vert \left ( D\,\widehat p\right )_{T^{C}}\right \Vert_{1}\leq 5\left \Vert \left (D(p^{*}-\widehat p)\right )_{T}\right \Vert_{1}
	+
	\frac{2\sqrt{\Vert \pi\circ p^{*}\Vert_{1}\,\Vert \pi^{-1}\Vert_{1}}\log\left (n\right )}{3n\rho \log\left (n\right ) }\;\Vert  \widehat{p}- p^{*}\Vert_{l_2(\pi)}.
\end{align}
Next, we need to provide a lower bound for $\left \Vert \mathfrak{M}\circ \left (p^{*}-\widehat{p}\right )\right \Vert^{2}_{2}$. 
It can be done on a set of vectors that satisfy  \eqref{eq:5_missing}  if $	\left \Vert p^{*}-\widehat{p} \right \Vert^{2}_{l_2(\pi)}$ is not too small. We start by introducing the corresponding set of vectors. Let 
\begin{equation}\label{def_Delta}
	\Delta=  \frac{\sqrt{2\Vert \pi\circ p^{*}\Vert_{1}\,\Vert \pi^{-1}\Vert_{1}}}{3n\rho}.
\end{equation}
For a set of indices $T$ and $\epsilon \geq \dfrac{\log(n)}{0.01\,\log\left (6/5\right )}$, we denote by $\mathcal{C}(T,\epsilon)$ the following set of vectors:
\begin{equation*}
	\mathcal{C}(T,\epsilon)=\left \{v\in[0,1]^{n} \,:\,  \left\Vert v\right\Vert_{l_2(\pi)}^{2}\geq \epsilon, \Vert \left (Dv\right) _{T^{C}}\Vert_{1}\leq 	5\Vert \left (Dv\right) _{T}\Vert_{1}+ \Delta \Vert  v\Vert_{l_2(\pi)}\right \}.
\end{equation*}
Next result provides partial isometry  for vectors in  $\mathcal{C}(T,\epsilon)$:
\begin{lemma}\label{lm:partial_isometry} With probability higher than $1-8/n$, we have for all $v\in \mathcal{C}(T, \epsilon)$
	$$\left \Vert \mathfrak{M}\circ v\right \Vert^{2}_{2}\geq \frac{\Vert v\Vert _{l_2(\pi)}^{2}}{2}-77 \left \{\dfrac{36\rho\sqrt{\kappa_{\pi} \vert T\vert }\log(n)}{\kappa_{T}}
	+\dfrac{\sqrt{\Vert \pi^{-1}\Vert_{1}}}{n}\left (5\sqrt{\Vert \pi\Vert_{1}}\log(n)+ 1\right) \right \}^{2}.$$
\end{lemma}
We now consider two cases, depending on whether the vector $p^{*}-\widehat{p}$ belongs to the set $ \mathcal{C}(T,\epsilon)$ or not.
\\
\textbf{Case 1}: Suppose first that $\left\Vert p^{*}-\widehat{p} \right\Vert_{l_2(\pi)}^{2} < \epsilon$, then the statement  of Theorem \ref{thm:denoising_missing} is true.

\medskip 

\noindent 
\textbf{Case 2}: It remains to consider the case  $\left\Vert p^{*}-\widehat{p} \right\Vert_{l_2(\pi)}^{2} \geq  \epsilon$. Let $T= \supp(Dp^*)$. Then,
\eqref{eq:5_missing} implies that  $p^{*}-\widehat{p} \in 	\mathcal{C}(T,\epsilon)$ with probability at least $1-4/n$. Then, combining \eqref{eq:4_missing} and Lemma \ref{lm:partial_isometry} we get
\begin{align}\label{eq:6_missing}
&\frac{1}{2}\Vert p^{*}-\widehat{p} \Vert^{2}_{l_2(\pi)}	\leq 77 \left \{\dfrac{36\rho\sqrt{\kappa_{\pi} \vert T\vert }\log(n)}{\kappa_{T}}
+\dfrac{\sqrt{\Vert \pi^{-1}\Vert_{1}}}{n}\left (5\sqrt{\Vert \pi\Vert_{1}}\log(n)+ 1\right) \right \}^{2}\nonumber\\
& + \frac{12\rho\,\sqrt{\kappa_{\pi}\vert T\vert} \log(n)}{\kappa_{T}}\left \Vert p^{*}-\widehat p\right \Vert_{l_2(\pi)}+
\frac{\sqrt{2\Vert \pi\circ p^{*}\Vert_{1}\,\Vert \pi^{-1}\Vert_{1}}\log\left (n\right )}{n}\;\Vert  \widehat{p}- p^{*}\Vert_{l_2(\pi)}.
\end{align}
Now the statement of Theorem \ref{thm:denoising_missing} follows applying 
$2ab\leq 8a^{2}+ \tfrac{1}{8}b^{2}$ to the last two terms in \eqref{eq:6_missing}.
\subsection{Proof of Lemma \ref{lm:partial_isometry}}\label{proof-lm:partial_isometre}
Set $\nu=\dfrac{\log(n)}{0.01\,\log\left (6/5\right )}$, $\alpha=\dfrac{6}{5}$ and
\[\mathcal E = 	77 \left \{\dfrac{36\rho\sqrt{\kappa_{\pi}\vert T\vert }\log(n)}{\kappa_{T}}
+\dfrac{\sqrt{\Vert \pi^{-1}\Vert_{1}}}{n}\left (5\sqrt{\Vert \pi\Vert_{1}}\log(n)+ 1\right) \right \}^{2}. \]
The proof of Lemma \ref{lm:partial_isometry} is based on the peeling argument. Denote by $\mathcal{B}$  the set that contains the complement of the event that we are interested in: 
\begin{equation*}
	\mathcal{B}=\left \{\exists\,v\in \mathcal{C}(T,\epsilon)\,\text{such that}\,\left \vert  \left\Vert \mathfrak {M}\circ v \right\Vert_{2}^{2}-\Vert v\Vert _{l_2(\pi)}^{2}\right \vert> \tfrac{1}{2}\Vert v\Vert _{l_2(\pi)}^{2}+ \mathcal{E}\right \}.
\end{equation*} 
For $r>\nu$ we will consider the following set of vectors:
$$\mathcal{C}(T,\epsilon,r)=\left \{v\in\mathcal{C}(T,\epsilon) \,:\,  \left\Vert v\right\Vert_{l_2(\pi)}^{2}\leq r \right \}
$$ 
and, 
for $l\in\mathbb N$, the following events
$$\mathcal{B}_l=\left \{\exists\,v\in \mathcal{C}(T,\epsilon,\alpha^{l}\nu)\,:\,\left \vert \left\Vert \mathfrak {M}\circ v \right\Vert_{2}^{2}-\Vert v\Vert _{L_2(\Pi)}^{2}\right \vert> \tfrac{5}{12}\alpha^{l}\nu+ \mathcal{E}\right \},$$
$$S_l=\left \{v\in \mathcal{C}(T,\epsilon)\,:\,\alpha^{l-1}\nu \leq \Vert v\Vert _{l_2(\pi)}^{2}\leq \alpha^{l}\nu\right \}.$$ 

  If the event $\mathcal{B}$ holds for some vector $v\in \mathcal{C}(T,\epsilon)$, then $v$ belongs to some $S_l$ and 
\begin{equation}\label{Bl}
	\begin{split}
		\left \vert \left\Vert \mathfrak {M}\circ v \right\Vert_{2}^{2}-\Vert v\Vert _{l_2(\pi)}^{2}\right \vert&> \tfrac{1}{2}\Vert v\Vert _{l_2(\pi)}^{2}+ \mathcal{E}\\&> \tfrac{1}{2}\alpha^{l-1}\nu+ \mathcal{E}\\&
		= \tfrac{5}{12}\alpha^{l}\nu+ \mathcal{E}
	\end{split}
\end{equation}
 which implies  $\mathcal{B}\subset\cup \,\mathcal{B}_l$. Next, we  apply the union bound using the following lemma: 
\begin{lemma}\label{lm:sup}
	Let
	$$Z_r=\underset{v\in \mathcal{C}(T,\epsilon,r)}{\sup}\left \vert \left\Vert \mathfrak {M}\circ v \right\Vert_{2}^{2}-\Vert v\Vert _{l_2(\pi)}^{2}\right \vert.$$ 
	We have that  \begin{equation*}
		pr\left ( Z_{r} \geq \frac{5}{12}r+ \mathcal{E}\right )\leq 4e^{-c_1\,r}
	\end{equation*}
	with $c_1\geq 0.01$.
\end{lemma}
Lemma \ref{lm:sup} implies that $\mathbb P\left (\mathcal{B}_l\right )\leq 4\exp(-c_1\,\alpha^{l}\nu)$. Using $e^{x}\geq x$, we obtain
\begin{equation*}
	\begin{split}
		\mathbb P\left (\mathcal{B}\right )&\leq \underset{l=1}{\overset{\infty}{\Sigma}}\mathbb P\left (\mathcal{B}_l\right )\leq 4\underset{l=1}{\overset{\infty}{\Sigma}}\exp(-c_1\,\alpha^{l}\nu)\\&\leq 4\underset{l=1}{\overset{\infty}{\Sigma}}\exp\left (-c_1\,\nu\,\log(\alpha)\,l\right )\leq \dfrac{4\exp\left (-c_1\,\nu\,\log(\alpha)\right )}{1-\exp\left (-c_1\,\nu\,\log(\alpha)\right )}
		\\&=\dfrac{4\exp\left (-\log(n)\right )}{1-\exp\left (-\log(n)\right )}.
	\end{split}
\end{equation*} 
This completes the proof of Lemma \ref{lm:partial_isometry}.
\subsection{Proof of Lemma \ref{lm:sup}}\label{pl1}

We first  provide an  upper bound on $\mathbb{E}(Z_r)$ and then show that $Z_r$ concentrates around its expectation. By definition of $Z_r$  we have
$$Z_r=\underset{v\in \mathcal{C}(T,\epsilon,r)}{\sup}\left \vert \sum _{i} m_iv_i^{2}-\bE\left (\sum _{i} m_iv_i^{2}\right )\right \vert.$$
 Using a standard symmetrization argument (see e.g. \citep{Ledoux_book}) we obtain 
\begin{equation*}
	\begin{split}
		\bE \left ( Z_r\right )&= \bE\left (\underset{v\in \mathcal{C}(T,\epsilon,r)}{\sup}\left \vert \sum_{i} m_{i}v_i^{2}-\bE\left (m_{i}v_i^{2}\right )\right \vert\right )\\&\leq 2\bE\left (\underset{v\in \mathcal{C}(T,\epsilon,r)}{\sup}\left \vert\sum_{i} \zeta_{i}m_{i}v_{i}^{2} \right \vert\right )
\end{split}\end{equation*}
where $\{\zeta_{i}\}$ is an i.i.d. Rademacher sequence. Then, using that $v_i\leq 1$, the contraction inequality (see e.g. \citep[Theorem 2.2]{Koltchinskii_st_flour}) yields
\begin{equation*}
	\bE \left ( Z_r\right )\leq 8\bE\left (\underset{v\in \mathcal{C}(T,\epsilon,r)}{\sup}\left \vert\sum_{i}\zeta_{i}m_{i}v_{i}\right \vert\right )=8\,\bE\left (\underset{v\in \mathcal{C}(T,\epsilon,r)}{\sup}\left \vert \left\langle \psi,v\right\rangle\right \vert\right )
\end{equation*}
where $\psi=(\zeta_{i}m_{i})_{i\in[n]}$. We can write
\begin{align}\label{lm_sup_eq:1}
	\left\langle \psi,v\right\rangle&= \langle \Pi\psi, v\rangle + \langle (I-\Pi)\psi  , v\rangle=
	\langle \Pi\psi  , \Pi v\rangle + \left\langle \left (D^{\dagger}\right )^{T} \psi  , D v\right \rangle
\end{align}
where $\Pi$ denote the projection on the kernel of $D$.
To bound the first term on the right hand side in \eqref{lm_sup_eq:1}, we first use the H\"older inequality:
\begin{align*}
	\left \vert \langle \Pi\psi , \Pi v\rangle \right \vert \leq \Vert \Pi\psi\Vert_{2}\Vert \Pi v\Vert_{2}.
\end{align*}
Using the same argument as in the proof of Theorem \ref{thm:denoising}, we have that 
$$\Vert \Pi\psi\Vert_{2}= \tfrac{1}{\sqrt{n}}\vert \langle \psi,\mathbf{1}_{n}\rangle\vert \quad \text{and}\quad \Vert \Pi v\Vert_{2}= \tfrac{1}{\sqrt{n}}\vert \langle v,\mathbf{1}_{n}\rangle\vert\enspace .$$
For $v\in \mathcal{C}(T,\epsilon,r)$ we get 
\begin{equation}\label{lm_sup_eq:2}
	\Vert \Pi v\Vert_{2}\leq \tfrac{1}{\sqrt{n}}\Vert v\Vert _{1}\leq \sqrt{\frac{\Vert \pi^{-1}\Vert_{1}}{n}}\left\Vert v\right\Vert_{l_2(\pi)}\leq \sqrt{\frac{\Vert \pi^{-1}\Vert_{1}\,r}{n}}\enspace .
\end{equation}
On the other hand,  Bernstein's inequality implies that 
that, with probability at least $1-\frac{1}{n}$, we have
\[\left \vert	\sum_{i=1}^{n}\psi_i\right \vert \leq \sqrt{2\Vert \pi\Vert_{1}}\log\left (2n\right )\]
and 
\begin{equation}\label{lm_sup_eq:3}
	\bE \left (\tfrac{1}{\sqrt{n}}\left \vert \langle \psi,\mathbf{1}_{n}\rangle \right \vert\right )\leq \sqrt{\frac{2\Vert \pi\Vert_{1}}{n}}\log\left (2n\right )+ \dfrac{1}{\sqrt{n}}
\end{equation}
where we used $\left \vert \langle \psi,\mathbf{1}_{n}\rangle \right \vert\leq n$.
Putting together \eqref{lm_sup_eq:2} and \eqref{lm_sup_eq:3} we get that for all 
\begin{align}\label{lm_sup_eq:9}
\mathbb{E}\left[	\sup_{v\in \mathcal{C}(T,\epsilon,r) }	\left \vert \langle \Pi\psi , \Pi v\rangle\right \vert\right]\leq \frac{\sqrt{\Vert \pi^{-1}\Vert_{1}\,r}}{n}
	\left (\sqrt{2\Vert \pi\Vert_{1}}\log\left (2n\right )+1\right ).
\end{align}
For the second term in \eqref{lm_sup_eq:1}, using again  the H\"older inequality, we get
\begin{align}\label{lm_sup_eq:4}
\left \vert	\left\langle \left (D^{\dagger}\right )^{T} \psi, Dv\right\rangle \right \vert &\leq \Vert \left (D^{\dagger}\right )^{T} \psi\Vert_{\infty}	\Vert D v\Vert_{1}.
\end{align}
Lemma \ref{lem:noise_term} implies that, with probability at least $1-4/n$, we have  
\begin{equation}\label{lm_sup_eq:5}
	\left \Vert \left (D^{\dagger}\right )^{T} \psi\right \Vert_{\infty}\leq  3\sqrt{2}\rho\log\left (8n\right ).
\end{equation}
On the other hand, we have that 
\[	\Vert \left (D^{\dagger}\right )^{T} \psi\Vert_{\infty}=\underset{j}{\sup}\left \vert	\sum_{i=1}^{n}D^{\dagger}_{ij}\psi_i \right \vert\leq 
\underset{j}{\sup} \sqrt{\sum_{i}\left (D^{\dagger}_{ij}\right )^{2}}\sqrt{\sum_{i}\psi_i^{2}}\leq 
\rho\sqrt{n}\]
which together with \eqref{lm_sup_eq:5} imply
\begin{equation}\label{lm_sup_eq:6}
	\bE\left (\left \Vert \left (D^{\dagger}\right )^{T} \psi\right \Vert_{\infty}\right )\leq 3\sqrt{2}\rho\log\left (n\right )+\dfrac{4\rho}{\sqrt{n}}\leq 6\rho\log\left (n\right )
\end{equation} 
for $n\geq 4$. Note that, as $v\in \mathcal{C}(T,\epsilon,r)$, we also obtain
	\begin{align}\label{lm_sup_eq:7}
		\Vert D v\Vert_{1}&\leq 	\Vert \left (D v\right )_{T}\Vert_{1}+ \Vert \left (D v\right )_{T^{C}}\Vert_{1}\leq 6 \Vert \left (D v\right )_{T}\Vert_{1}+\Delta \Vert  v\Vert_{l_{2}(\pi)}\nonumber\\
		&\leq \dfrac{6\sqrt{\vert T\vert }}{\kappa_{T}}\Vert v\Vert_{2}+\Delta \sqrt{r}
		\nonumber\\
		&\leq \dfrac{6\sqrt{\kappa_{\pi}\vert T\vert }}{\kappa_{T}} \Vert  v\Vert_{l_{2}(\pi)}+\Delta \sqrt{r}\quad \text{using \eqref{relation_norms}}\nonumber\\
		&\leq \left (\dfrac{6\sqrt{\kappa_{\pi}\vert T\vert }}{\kappa_{T}} +\Delta\right ) \sqrt{r}.
	\end{align}	
Using \eqref{lm_sup_eq:6} and \eqref{lm_sup_eq:7} we get 
\begin{align}\label{lm_sup_eq:8}
	\bE \left (\underset{v\in \mathcal{C}(T,\epsilon,r)}{\sup} \left\langle \left (D^{\dagger}\right )^{T} \psi, Dv\right \rangle\right )\leq 
6\rho\sqrt{r}\log (n)	\left (\dfrac{6\sqrt{\kappa_{\pi}\vert T\vert }}{\kappa_{T}}+\Delta\right ).
\end{align}
Putting together \eqref{lm_sup_eq:9}, \eqref{lm_sup_eq:8} and using the definition~\eqref{def_Delta} of $\Delta$,  we get
\begin{equation*}
	\bE \left ( Z_r\right )\leq 8\sqrt{r}\left \{\dfrac{36\rho\sqrt{\kappa_{\pi}\vert T\vert }\log(n)}{\kappa_{T}}
	+\dfrac{\sqrt{\Vert \pi^{-1}\Vert_{1}}}{n}\left (5\sqrt{\Vert \pi\Vert_{1}}\log(n)+ 1\right) \right \}
\end{equation*}
where we used $\Vert \pi\circ p^{*}\Vert_{1}\leq \Vert \pi\Vert_{1}$. Now, $2ab\leq a^{2}+b^{2}$ implies
\begin{align}\label{bound_expectation_sup}
	\bE \left ( Z_r\right )&\leq \dfrac{5}{24}r+ 77\left \{\dfrac{36\rho\sqrt{\kappa_{\pi}\vert T\vert }\log(n)}{\kappa_{T}}
	+\dfrac{\sqrt{\Vert \pi^{-1}\Vert_{1}}}{n}\left (5\sqrt{\Vert \pi\Vert_{1}}\log(n)+ 1\right) \right \}^{2}\nonumber
	\\ &
	\hskip 0.5 cm  =  \dfrac{5}{24}r+\mathcal{E}
\end{align}
Now we can show that $Z_r$ concentrates around its expectation using  Talagrand's concentration inequality \citep{Talagrand_1996}, which in the current  can be obtained by inverting the tail bound in Theorem 3.3.16 in \citep{Ginenickl}. 
\begin{theorem} \label{Talagrand} Let $(S, \mathcal{S})$ be a measurable space and let $n \in \mathbb{N}$. 
	Let $X_k,~k=1,\dots,n  $ be independent ${S}$-valued random variables and let $\mathcal{F}$ be a countable set of functions $f=(f_1,...,f_n):{S}^n \rightarrow [-K,K]^n$ such that $\mathbb{E}f_k(X_k)=0$ for all $f\in \mathcal{F}$ and $k=1,...,n$. Set
	$$Z:=\sup_{f \in \mathcal{F}} \sum_{k=1}^n f_k(X_k)\enspace  .$$
	Define the variance proxy
	$$V_n:=2K\mathbb{E}Z + \sup_{f \in \mathcal{F} }\sum_{k=1}^n \mathbb{E} \left [(f_k(X_k))^2 \right ].$$
	Then, for all $t \geq 0$,
	\begin{equation*}
		pr\left ( Z - \mathbb{E} Z \geq t \right ) \leq \exp \left ( \frac{-t^2}{4V_n + (9/2)Kt} \right ).
	\end{equation*}
\end{theorem}
Note that the functional $v \to \left\Vert \mathfrak {M}\circ v \right\Vert_{2}^{2}-\Vert v\Vert _{l_2(\pi)}^{2}$ is continuous on the set of vectors   $\mathcal{C}(T,\epsilon,r) $. Hence it's enough to consider a dense countable subset of $\mathcal{C}(T,\epsilon,r) $ and we may apply Talagrand's inequality to $Z_r$. We have for our particular case,  since  $\sup_{f \in \mathcal{F}} |f(X)|=\sup_{f \in \{ \mathcal{F} \bigcup\{ -\mathcal{F}\}\}} f(X)$, that
$$X_{i}=m_{i}-\mathbb{E}\left ( m_{i}\right ), ~~~S=[-1,1].$$
We can compute:
\begin{align*}
\sum_{i}\mathbb{E} \left [(f_k(X_k))^2 \right ]=\sum_{i}\mathbb{E} \left [ m_{i}-\mathbb{E}\left ( m_{i}\right )\right ]^{2}v^{4}_{i}
\leq \sum_{i} \pi_{i} v^{2}_{i}=\Vert v\Vert _{l_2(\pi)}^{2}\leq r
\end{align*}
which implies that 
$V_{n}\leq \tfrac{17}{12}r+2\mathcal{E}$.
Hence, using \eqref{bound_expectation_sup} and
Talagrand's inequality with $t=\dfrac{5}{24}r+0.5\,\mathcal{E}$, we obtain
\begin{align*}
	 	pr\left ( Z_{r} > \frac{5}{12}r+1.5\mathcal{E}\right ) \leq \exp (-c_1r )
\end{align*}
	with $c_1\geq 0.01$ which completes the proof of Lemma \ref{lm:sup}.

\appendixfour
\section{Further Results: Simulation study}\label{appendix:results}
In this section, we report additional results of our simulation study.

\subsection{Algorithms and Compute time}\label{sec:app:solver}

In our simulation study, we use the semismooth Newton based augmented Lagrangian
method (SSNAL) proposed by \cite{sun2021convex}.  SSNAL is a scalable algorithm that was originally designed to perform convex clustering to solve the following optimization problem:

 $$ \min_{X \in \R^{d \times n}} \frac{1}{2} \sum_{i=1}^n \| x_i - a_i\|^2 + \lambda \sum_{i < j } W_{ij} \| x_i -x_j\|_2 $$
Note that in our case, since our vector $p$ is one dimensional $\| p_i -p_j\|_2 = \sqrt{ (p_i - p_j)^2 } = | p_i - p_j|$, and our problem shares the same objective function as convex clustering.
In their work, \cite{sun2021convex} prove their algorithm to be theoretically efficient and practically efficient and robust. More specifically, the authors show that, as long as the dimension of the feature vector is small (i.e, in our case, equal to 1), this method has the potential of considerably improving upon existing solvers.

\xhdr{Evaluation} 
To highlight the scalability of SSNAL, we propose comparing here four different solvers on a set of synthetic experiments. More specifically, we compare the performances of:
\begin{enumerate}
   \item{\it An off-the-shelf convex optimization solver (CVXPY)},
   \item{\it The dual-coordinate descent algorithm} proposed by Tibshirani and Taylor~\cite{tibshirani2011solution}. This method relies on solving the dual problem:
$$ \text{minimize} \frac{1}{2}  \| y - \Gamma^T u \|^2 \text{  subject to } \| u \|_{\infty}\leq \lambda,$$
yielding a simple algorithm.
   \item{\it An ADMM-based approach}: ADMM is often the method of choice to deal with large scale constrained and has been suggested in multiple instance as the method of choice for graph-based penalties on large graphs~\cite{hallac},
   \item{\it The SSNAL algorithm of \cite{sun2021convex},} described in the previous paragraph.
\end{enumerate}

In this first set of experiments, we propose comparing the performance of the different algorithms when the underlying is an Erdos Renyi random graph. The reason for considering this topology here is that this is a difficult setting, since the graph is quite dense. Since the complexity of the penalty depends on the number of edges, we plot the time requirements of the algorithm as the number of nodes increases and the penalty varies.
The results are displayed in Figure~\ref{fig:compute_time} below.

\begin{figure}[h]
    \centering
    \includegraphics[width=\textwidth]{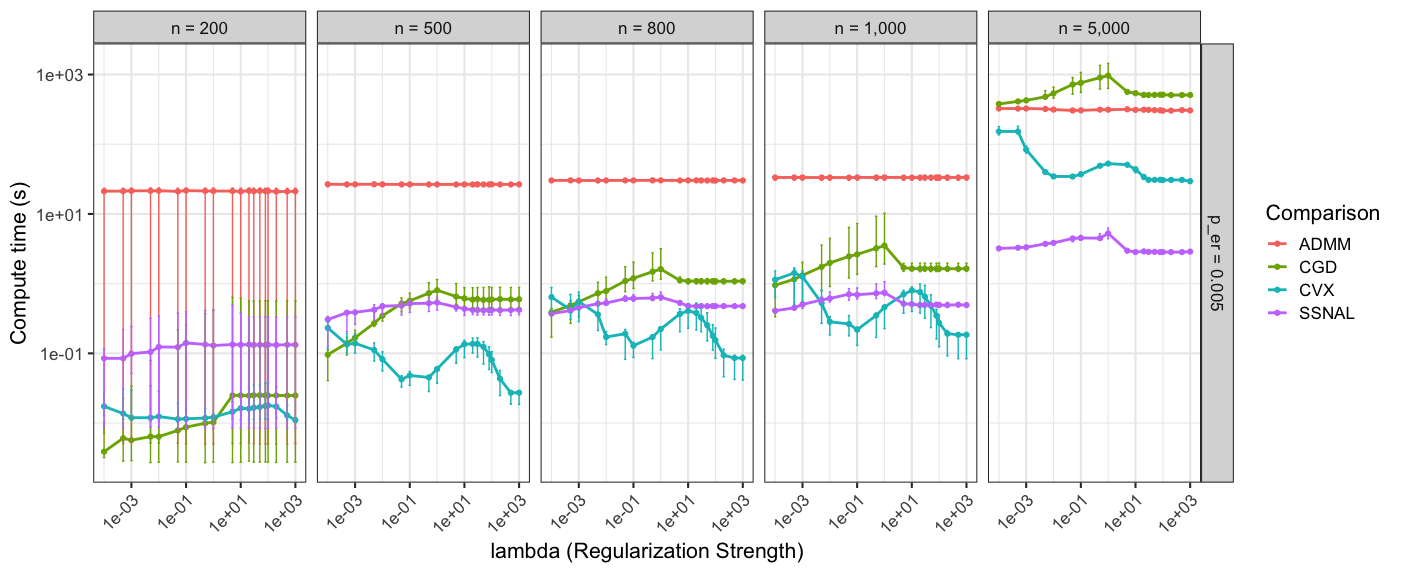}
    \caption{Computation time for the different algorithms}
    \label{fig:compute_time}
\end{figure}

\subsection{Synthetic Experiments: Additional Results}
\begin{figure}
    \centering
    \includegraphics[width=\textwidth]{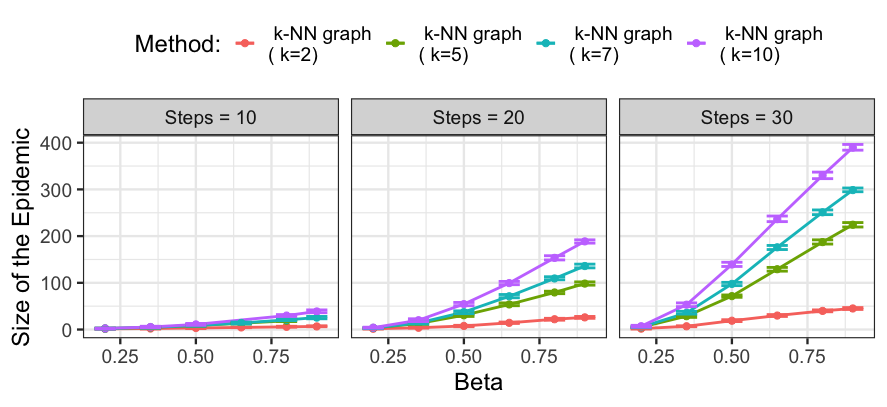}
    \caption{Number of infections for the kNN graph as a function of $\beta$, $n_{\text{steps}}$ and the parameter $k$ in the k-NN graph. The healing rate is here fixed to $\gamma=0.1.$ Results are averaged over 100 independent experiments.}
    \label{fig:knn_infection}
\end{figure}


\begin{figure}
    \centering
    \includegraphics[width=\textwidth]{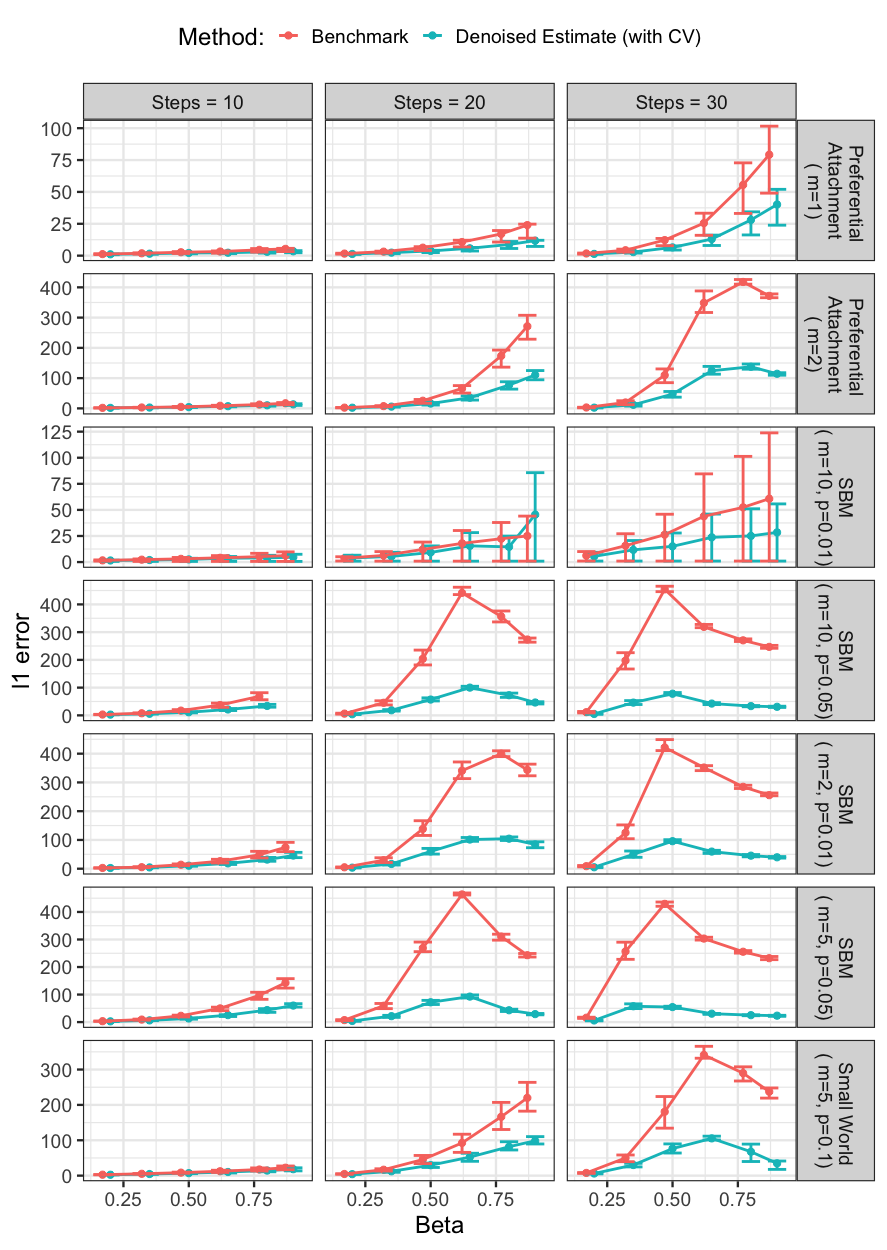}
    \caption{$\ell_1$ error for various graph topologies as a function of $\beta$, $n_{\text{steps}}$. The healing rate is here fixed to $\gamma=0.1.$ Results are averaged over 100 independent experiments, and error bars denote interquartile ranges.}
    \label{fig:resothers}
\end{figure}

\begin{figure}
    \centering
    \includegraphics[width=\textwidth]{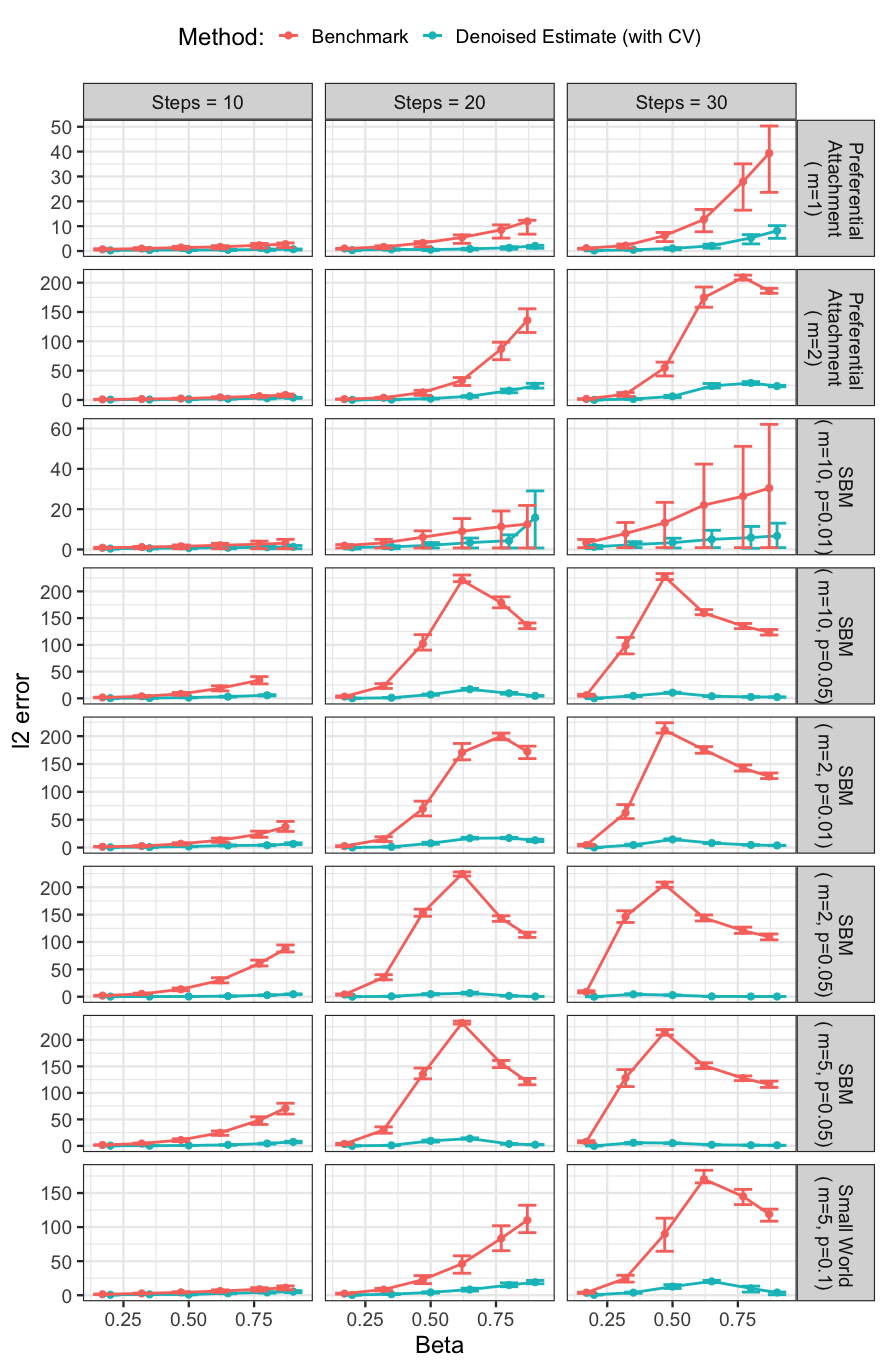}
    \caption{$\ell_2$ error for various graph topologies as a function of $\beta$, $n_{\text{steps}}$. The healing rate is here fixed to $\gamma=0.1.$ Results are averaged over 100 independent experiments, and error bars denote interquartile ranges.}
     \label{fig:resother2}
\end{figure}

\begin{figure}
    \centering
    \includegraphics[width=\textwidth]{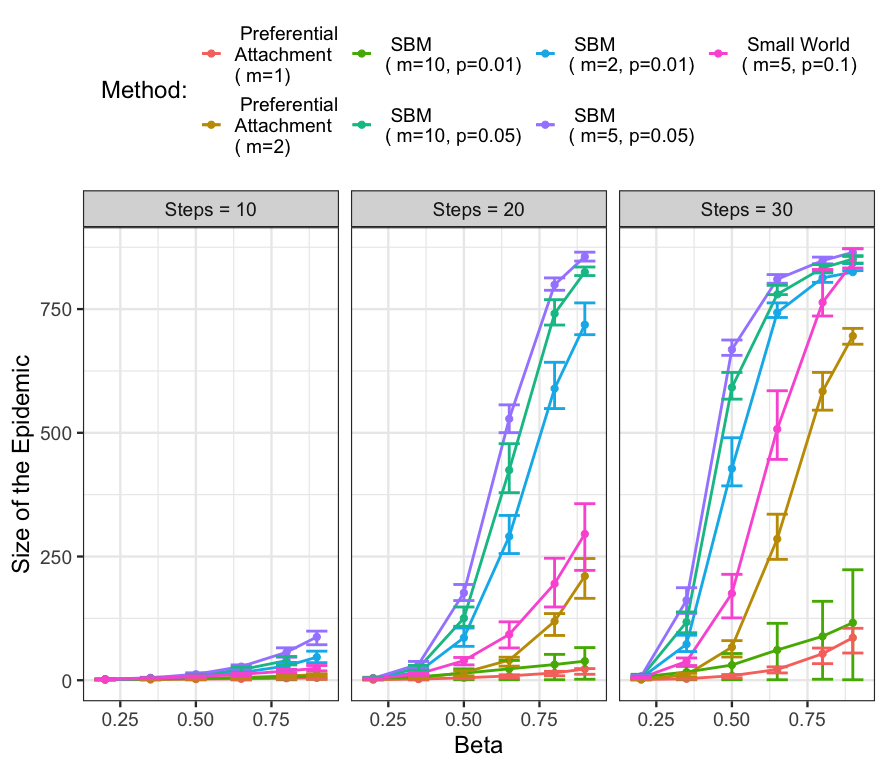}
    \caption{Number of infections for graphs with various topologies as a function of $\beta$, $n_{\text{steps}}$. The healing rate is here fixed to $\gamma=0.1.$ Results are averaged over 100 independent experiments, and error bars denote interquartile ranges.}
    \label{fig:resothers_inf}
\end{figure}




\bibliographystyle{biometrika}
\bibliography{Diffusion_Networks}

\end{document}